\documentclass[a4paper, 10pt]{article}

\title{Strategy-proof voting schemes under single-peaked preferences over intervals
}
\author{Hao Yu\thanks{Adam Smith Business School, University of Glasgow, United Kingdom. Email: \href{mailto:hao.yu@glasgow.ac.uk}{hao.yu@glasgow.ac.uk}}
}

\date{\today}

\usepackage[margin=1.4in]{geometry}
\selectfont
\usepackage{amsmath, amsfonts, amsthm, amssymb, graphicx, enumitem}
\usepackage[colorlinks=true, citecolor=blue]{hyperref}
\usepackage[round]{natbib}

\usepackage{subcaption}
\usepackage{tikz}
\usetikzlibrary{patterns.meta, fit}

\usepackage{newpxtext, newpxmath}

\newtheorem{theorem}{Theorem}
\newtheorem{proposition}{Proposition}

\newtheorem{lemma}{Lemma}

\theoremstyle{definition}

\newtheorem{inneraxiom}{Axiom}
\newtheorem*{inneraxiomnonumber}{\axiomname}
\providecommand{\axiomname}{}

\makeatletter
\NewDocumentEnvironment{axiom}{o}
{\IfNoValueTF{#1}
    {\inneraxiom}
    {
    \renewcommand{\axiomname}{#1}
    \def\@currentlabel{#1}
    \inneraxiomnonumber
    }%
}
{\IfNoValueTF{#1}
    {\endinneraxiom}{\endinneraxiomnonumber}
}
\makeatother

\newtheorem{example}{Example}
\AtBeginEnvironment{example}{%
  \pushQED{\qed}%
}
\AtEndEnvironment{example}{\popQED}

\newtheoremstyle{no-bracket}
    {\topsep}
    {\topsep}
    {\normalfont\itshape}
    {0pt}
    {\bfseries}
    {} 
    { }
    {\thmname{#1}~\thmnumber{#2}\thmnote{#3}.}
\theoremstyle{no-bracket}  
\newtheorem*{theorem*}{Theorem}

\usepackage{lipsum}

\newcommand{\med}{\mathop{\mathrm{med}}}
\newcommand{\rank}{\mathop{\mathrm{rank}}}
\newcommand{\calI}{\mathcal{I}}
\newcommand{\medI}{\mathop{M}}
\newcommand{\leftI}{\mathop{L}}
\newcommand{\rightI}{\mathop{R}}
\newcommand{\rankI}{\mathop{K}}

\begin{document}

\maketitle
\begin{abstract}
    When a subset of alternatives is allowed to be the collective choice, one approach is to study the social choice correspondence defined on agents' preferences over alternatives. This paper considers a different approach. We argue that agents should have preferences directly over outcomes. We consider the social choice problem on a closed real interval, and assume a single-peaked domain over intervals. Our focus is on the peak-only social choice function, referred to as the voting scheme. The leftright scheme, which generalizes the minmax scheme, is the only strategy-proof voting scheme. Adding anonymity, it reduces to the generalized median scheme. For an odd number of agents, the median scheme is the only voting scheme satisfying efficiency, anonymity, and strategy-proofness.
    \\
    
    \noindent\textbf{Keywords:} Social Choice Function; Peak-onlyness; Single-peakedness; Interval; Strategy-proofness; The median scheme.
    \\[0.2cm]
    \noindent\textbf{JEL Classification:} D71 
\end{abstract}

\newpage
\section{Introduction}

This paper considers the social choice problem of choosing a subinterval from within an interval. In many scenarios, people are required to reach a collective decision by aggregating different opinions over intervals. For example, experts need to combine their forecasts, which could be confidence intervals derived from rigorous statistical methods or intervals based on subjective evaluations \citep{clemen1989}. 
In the context of legal systems, minimum or maximum limits of tolerance are, in many cases, determined by community standards \citep{miller2025}. Another interpretation is the political spectrum, and each interval outcome can be viewed as a connected coalition \citep{block2010}. 

We use the collective choice of an event time as a running example. The event could be a festival celebration or an academic conference, and the event time specifies when it starts and ends. Splitting the event time is not allowed. The organizer asks participants about their preferences over the event times and makes a collective decision based on their messages. The organizer hopes that participants will report truthfully so that the collective decision can reflect the consensus to some extent. 

That leads to our special focus on strategy-proofness.\footnote{For surveys, see \cite{sprumont1995} and \cite{barbera2011}.} The Gibbard-Satterthwaite theorem indicates an impossibility, that is, the only strategy-proof social choice function is dictatorial \citep{Gibbard1973, Satterthwaite1975}. A similar impossibility result still holds when the subset of alternatives is allowed to be the outcome \citep*{barbera2001}. 

To achieve a possibility result, a common approach is to consider a domain restriction, where the single-peaked domain has been extensively successful. In a one-dimensional model, \cite{moulin1980} characterizes the class of strategy-proof voting schemes, later called minmax schemes, and the generalized median scheme. For a social choice among $n$ agents, the generalized median scheme selects the median opinion among $n$ agents and $n+1$ fixed ballots; and it is efficient if the number of fixed ballots is restricted to $n-1$. 
This result is then extended to multidimensional models with appropriate generalizations of the traditional single-peaked domain \citep*{border1983, barbera1993}. 

In terms of the interval model, a branch of the literature, e.g., \cite{klaus2020}, considers an approach of extension. They assume that each agent has a traditional single-peaked preference over alternatives, and infers a best-worst extension of preferences over intervals. Viewing the result of social choice correspondence as the outcome under ignorance, this preference extension achieves a balance between optimistic and pessimistic opinions.

We argue that if the outcome is allowed to be a subset of alternatives, agents should have direct preferences over them \citep*{barbera1991}. Especially, in our running example, we can say that the agent prefers the event running from 8:00 on day 1 to 17:00 on day 3; but can hardly say that each agent prefers only one specific time, e.g., at 13:00 on October the 1st. A similar extension doesn't work in our situation.

We believe that single-peaked preferences in the following sense are an appropriate setting \citep{block2010,farfel2011}. Each agent has a favorite event time, e.g., from 10:00 on day 1 to 17:00 on day 3. Consider another event time, e.g., from 17:00 on day 1 to 8:00 on day 3. Every event time, e.g., from 12:00 on day 1 to 12:00 on day 3, of which the start and end times are between those of the two schedules, respectively, is preferred to the fixed one. We refer to it as the single-peaked preference over intervals. 
Note that we can always represent each feasible outcome, i.e., an interval containing leftmost and rightmost alternatives, as a point in two-dimensional space. This is, in fact, a two-dimensional generalized single-peaked preference proposed by \cite*{barbera1993}.

We restrict our attention to the peak-only social choice functions, referred to as voting schemes. That is, the social choice depends only on the profile of peak intervals. There are two reasons for this restriction. First, reporting the peak interval is more practical than reporting full preference. Suppose there are $m$ alternatives in total. Reporting a (strict) ordering of all outcomes would require a list of $m^2/2$ intervals, which is more demanding than that of selecting a single alternative. Second, dropping peak-onlyness would not provide more insight into the form of the social choice function. All schemes characterized in this paper only need the information of peak intervals. We will further discuss the peak-onlyness at the end. 

In a multidimensional setting, \cite*{barbera1993} show that every onto social choice function satisfying strategy-proofness is a generalized median voter scheme, which generalizes \cite{moulin1980}'s schemes, on the generalized single-peaked domain. This scheme is induced and represented by a left (right) coalition system. This result cannot be applied directly to our interval model. This is not because their model is discrete, rather, we face an additional feasibility constraint, i.e., the leftmost alternative should not be to the right of the rightmost alternative. \cite*{barbera1997} consider the case where the alternative set may be non-Cartesian. They show that strategy-proof social choice functions are similar, though the forms are restricted. 

We characterize the class of voting schemes satisfying strategy-proofness, which we call the leftright scheme. This scheme has the same structure as the one-dimensional minmax scheme characterized by \cite{moulin1980}. Importantly, it is not merely a combination of two minmax schemes for leftmost and rightmost alternatives, respectively; the fixed ballots for all coalitions are also intervals, so that the feasibility constraints are automatically satisfied. It provides an alternative representation of the generalized median voter scheme. In addition to the leftright scheme, we characterize the interval version of the generalized median scheme and its restrictions under anonymity, unanimity, and ontoness. 

We also consider efficiency. Establishing the possibility result of strategy-proofness and efficiency is difficult in a multidimensional setting, especially in settings with more than two dimensions \citep{barbera1991,kim1984,nehring2007b,peters1992}. In particular, \cite{nehring2007b} study strategy-proof and efficient voting schemes on a median space, and show that for more than two dimensions, it is almost dictatorial. 

Identifying the efficient and strategy-proof social choice functions in our interval model, which has two dimensions, is therefore meaningful. 

In a discrete model, the single-peaked domain considered in our paper forms a median space. \cite{nehring2007} show that there exists an anonymous, strategy-proof, and weakly neutral social choice function. In a discrete interval model, \cite{block2010} finds that the median scheme, of which the leftmost and rightmost alternatives are the medians of leftmost and rightmost alternatives among agents' peak intervals, respectively, is an example. 
Using the result of \cite{nehring2007b}, \cite{block2010} shows that it is also efficient, and thus provides a possibility result.

Different from \cite{block2010}, the intervals we consider are closed subsets of $\mathbb{R}$. Based on this setting, we provide a novel analysis and show that the median scheme is the only efficient, anonymous, and strategy-proof voting scheme. However, our approach does not apply to the discrete model. We provide a discussion on it at the end.

In a similar setting, \cite{miller2025} provides several characterizations of the endpoint scheme, which is a subclass of the generalized median scheme, by using a weak version of neutrality. It selects an interval of which the leftmost and rightmost alternatives are from those of agents' peak intervals, respectively. This work is also related to ours. Indeed, the rank-based scheme and the median scheme discussed in our paper belong to the endpoint scheme. 

The rest of the paper is organized as follows: Section~\ref{sec:ind2} describes the model and preliminaries; and introduces the voting schemes. Section~\ref{sec:ind3} studies strategy-proofness and provides its characterization. Section~\ref{sec:ind4} gives further analysis, including the characterization results regarding to anonymity, unanimity, ontoness, and efficiency. All the voting schemes characterized in this section have a similar median structure. Section~\ref{sec:ind5} concludes and discusses the paper.

\section{The model}\label{sec:ind2}
Let $A$ denote the set of alternatives, a closed interval on $\mathbb{R}$.
For non-triviality, we assume that it is not a singleton. We use $[\cdot]^-$ and $[\cdot]^+$ to denote the leftmost and rightmost alternatives of the corresponding interval, respectively. For example, $A^-$ and $A^+$ denote the leftmost and the rightmost alternatives of $A$, respectively; and $A\equiv[A^-, A^+]$. 
Denote by $\calI = \{[l,r]\mid l,r\in A, l\leq r\}$ the set of all feasible (closed) intervals generated by the alternative set $A$.

For each pair of intervals $X, Y\in\calI$, we define the betweenness set as 
\begin{equation*}
    B(X,Y) 
    = 
    \{
    V\in\calI
    \mid 
    \min\{X^-, Y^-\}\leq V^- \leq \max\{X^-, Y^-\}, 
    \min\{X^+, Y^+\}\leq V^+ \leq \max\{X^+, Y^+\} 
    \}.
\end{equation*}
That is, we say $V$ is between $X$ and $Y$ if its leftmost and rightmost alternatives are between the leftmost and rightmost alternatives of $X$ and $Y$, respectively. 

We also define the order relation between two intervals. For each pair $X, Y\in\calI$, we abuse the notation and write $X\leq Y$ if $X^-\leq Y^-$ and $X^+\leq Y^+$. Note that each interval $V$ can be written as $[V^-, V^+]$. $X\leq Y$ is exactly the same as the order relation between the two-dimensional coordinates $(X^-, X^+)$ and $(Y^-, Y^+)$.

Let $N= \{1, \cdots, n\}$ denote the set of $n$ agents. Denote by $\succsim_i$ the rational, i.e., complete and transitive, preference relation of agent $i\in N$ over intervals $\calI$. For each pair $X, Y\in\calI$, we say that agent $i$ strictly prefers $X$ to $Y$, written as $X\succ_i Y$, if $X\succsim_i Y$ but not $Y\succsim_i X$. 

We assume agents' preferences over intervals $\calI$ are single-peaked. That is, (a) for all $i\in N$, there is a unique peak interval $X_i^*$ such that for all $Y\in \calI:Y\neq X_i^*$, we have $X_i^*\succ_i Y$; and (b) for all $Y,Z\in\calI$ such that $Z\in B(X_i^*, Y)$ and $Z\neq Y$, we have $Z\succ_i Y$.
Denote by $\mathcal{P}$ the set of all preferences satisfying the single-peakedness properties shown above. We refer to it as the single-peaked domain over intervals.

This generalization of traditional single-peaked preferences is proposed by \cite*{barbera1993}, and introduced in the interval model by \cite{block2010}, \cite{farfel2011}. Figure~\ref{fig:ind5} illustrates this preference in both interval and two-dimensional representations. This single-peaked domain includes many common preference classes such as additively separable preferences and distance-based preferences. 

\begin{example}
    An $L_p$ (dis)utility, $1\leq p<\infty$, belongs to the single-peaked domain.

    Suppose $X\in\calI$ is a peak. The $L_p$ utility function  $u:\calI\to\mathbb{R}$ with peak $X$ is defined as
\begin{equation*}
    u^{(p)}(V)
    =
    - |V^- - X^-|^p + |V^+ - X^+|^p.
\end{equation*}
    In particular, we write $u^{(2)}(V) = -||V-X||^2$, where $||V-X|| \equiv \sqrt{|V^- - X^-|^2 + |V^+ - X^+|^2}$.
\end{example}

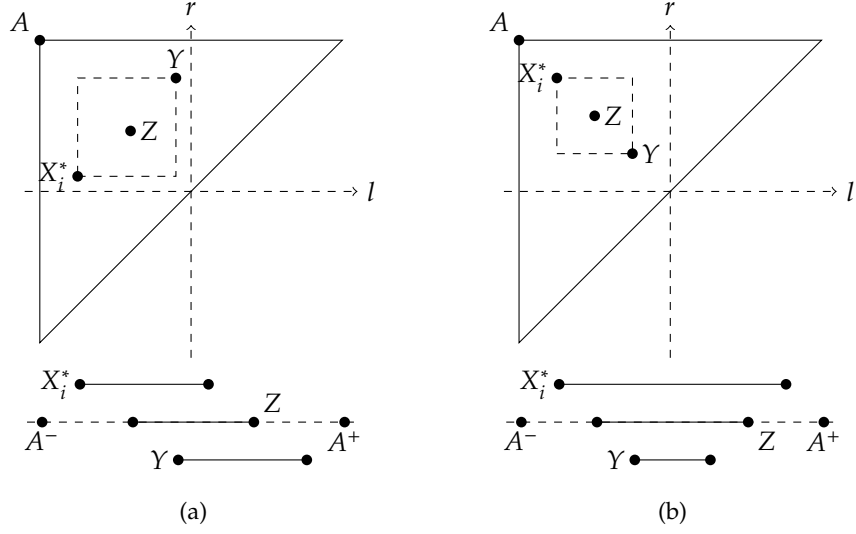
\begin{figure}
    \centering
\begin{subfigure}[t]{0.45\textwidth}
\centering
\begin{tikzpicture}[scale=1] 
    \draw[->, dashed] (-2.2,0) -- (2.2,0) node[right] {$l$}; 
    \draw[->, dashed] (0,-2.2) -- (0,2.2) node[above] {$r$}; 

    \draw[-] (-2,-2) -- (2,2) -- (-2, 2) -- cycle;

    \fill (-1.5,0.2) circle (2pt); \node[left] at (-1.5,0.2) {$X_i^*$}; 
    \fill (-0.8,0.8) circle (2pt); \node[right] at (-0.8,0.8) {$Z$}; 
    \fill (-0.2,1.5) circle (2pt); \node[above] at (-0.2,1.5) {$Y$}; 

    \fill (-2, 2) circle (2pt); \node[above left] at (-2, 2) {$A$}; 

    \coordinate (X) at (-1.5,0.2); \coordinate (Y) at (-0.2,1.5);
    \node[ draw, dashed,  inner sep=0pt, fit=(X)(Y) ] {};
\end{tikzpicture}

\begin{tikzpicture}[scale=1] 
    \draw[dashed] (-2.2,0) -- (2.2,0);
    \fill (-2,0) circle (2pt) node[below] {$A^-$}; \fill (2,0) circle (2pt) node[below] {$A^+$};
    
    \draw[-] (-0.8, 0)--(0.8, 0) node[above right] {$Z$};
    \fill (-0.8,0) circle (2pt); \fill (0.8,0) circle (2pt); 

    \draw[-] (0.2, 0.5)--(-1.5, 0.5) node[left] {$X_i^*$};
    \fill (0.2, 0.5) circle (2pt); \fill (-1.5, 0.5) circle (2pt); 

    \draw[-] (1.5, -0.5) -- (-0.2, -0.5) node[left] {$Y$};
    \fill (-0.2, -0.5) circle (2pt); \fill (1.5, -0.5) circle (2pt); 
\end{tikzpicture}
\subcaption{}
\end{subfigure}
\begin{subfigure}[t]{0.45\textwidth}
\centering
\begin{tikzpicture}[scale=1] 
    \draw[->, dashed] (-2.2,0) -- (2.2,0) node[right] {$l$}; 
    \draw[->, dashed] (0,-2.2) -- (0,2.2) node[above] {$r$}; 

    \draw[-] (-2,-2) -- (2,2) -- (-2, 2) -- cycle;

    \fill (-1.5,1.5) circle (2pt); \node[left] at (-1.5,1.5) {$X_i^*$}; 
    \fill (-1,1) circle (2pt); \node[right] at (-1,1) {$Z$}; 
    \fill (-0.5,0.5) circle (2pt); \node[right] at (-0.5,0.5) {$Y$}; 

    \fill (-2, 2) circle (2pt); \node[above left] at (-2, 2) {$A$}; 

    \coordinate (X) at (-1.5,1.5); \coordinate (Y) at (-0.5,0.5);
    \node[ draw, dashed,  inner sep=0pt, fit=(X)(Y) ] {};
\end{tikzpicture}

\begin{tikzpicture}[scale=1] 
    \draw[dashed] (-2.2,0) -- (2.2,0);
    \fill (-2,0) circle (2pt) node[below] {$A^-$}; \fill (2,0) circle (2pt) node[below] {$A^+$};
    
    \draw[-] (-1, 0)--(1, 0) node[below right] {$Z$};
    \fill (-1,0) circle (2pt); \fill (1,0) circle (2pt); 

    \draw[-] (1.5, 0.5)--(-1.5, 0.5) node[left] {$X_i^*$};
    \fill (1.5, 0.5) circle (2pt); \fill (-1.5, 0.5) circle (2pt); 

    \draw[-] (0.5, -0.5) -- (-0.5, -0.5) node[left] {$Y$};
    \fill (-0.5, -0.5) circle (2pt); \fill (0.5, -0.5) circle (2pt); 
\end{tikzpicture}
\subcaption{}
\end{subfigure}
    \caption{Two examples of $Z\in B(X_i^*,Y)$ with interval and two-dimensional representations. }
    \label{fig:ind5}
\end{figure}

On the domain of single-peaked preference over intervals, we have the following richness condition.

\begin{axiom}[Richness condition]
    For all $X, Y, Z\in\calI$ such that $Z\notin B(X, Y)$, there exists a preference $\succsim\ \in\mathcal{P}$ with peak $X$ such that $Y\succ Z$.
\end{axiom}

Note that the single-peaked domain defined above contains those satisfying continuity. We can always find some continuous preferences represented by utility functions to establish this richness condition. 

\begin{example}
    Suppose $Y \leq X$ and $Z\notin B(X,Y)$. Note that the condition $Z\notin B(X,Y)$ implies that $Z\neq X$ and $Z\neq Y$. For each nontrivial case, we construct a utility function $u:\calI\to\mathbb{R}$ with peak $X$ such that $u(Y)>u(Z)$.

    Case 1: $Z\leq Y$. Immediately, we obtain $Y\in B(X, Z)$. By single-peakedness, we have $Y\succ Z$. 
    
    Case 2(a): $Y^-<Z^-\leq X^-$ and $Z^+< Y^+$. Let
\begin{equation*}
    u(V) = -\lambda (Y^+ - Z^+) |X^- - V^-| - (Z^- - Y^-)|X^+ - V^+|,
\end{equation*}
    for some $0<\lambda<1$. Then we have
\begin{equation*}
\begin{split}
    u(Y) - u(Z)
    &=
    -\lambda (Y^+ - Z^+)(|X^- - Y^-| - |X^- - Z^-|) - (Z^- - Y^-)(|X^+ - Y^+|- |X^+ - Z^+|) \\
    &=
    -\lambda (Y^+ - Z^+)(Z^- - Y^-) - (Z^- - Y^-)(Z^+ - Y^+) \\
    &= 
    (1-\lambda) (Y^+ - Z^+)(Z^- - Y^-)
    > 0.
\end{split}
\end{equation*}
    
    Case 2(b): $Z^-< Y^-$ and $Y^+<Z^+\leq X^+$. A symmetric argument establishes $Y\succ Z$ with the following utility function. 
\begin{equation*}
    u(V) = - \lambda (Y^- - Z^-) |X^+ - V^+| - (Z^+ - Y^+)|X^- - V^-|,
\end{equation*}
    for some $0<\lambda<1$. 

    Case 3: $Z^- > X^-$ or $Z^+ > X^+$. Let 
\begin{equation*}
    u(V)
    =
    \left\{
        \begin{array}{ll}
            -\dfrac{||X-Z||}{||X-Y||+\varepsilon}||X - V|| & \text{if $V\leq X$}; \\
            -||X - V||, & \text{otherwise},
        \end{array}
    \right.
\end{equation*}
    for some $\varepsilon > 0$. Then 
\begin{equation*}
    u(Y) - u(Z)
    =
    ||X - Z|| - \frac{||X-Z||}{||X-Y||+\varepsilon}||X - Y||
    =
    ||X - Z|| \left(1 - \frac{||X-Y||}{||X-Y||+\varepsilon}\right)
    =
    \frac{\varepsilon ||X - Z||}{||X-Y||+\varepsilon}
    >0,
\end{equation*}
    as desired.
\end{example}

\subsection{Operators and lemmata}

To assist our analysis, we define some operators for alternatives (scalars) and intervals.

For scalars, consider a vector $v = (v_1, \cdots, v_n)\in A^n$ such that $v_1\leq \cdots \leq v_n$. We denote by $\rank_\gamma(\cdot)$ the rank-based operator with parameter $\gamma$. It selects the $\gamma$-th smallest scalar $v_\gamma$ among $v$.
Fix an odd $n$. The median operator, denoted by $\med(\cdot)$, selects the median among $v_1, \cdots, v_n$. That is, $|\{v_j\in v: v_j\leq \med(v)\}|\geq (n+1)/2$ and $|\{v_j\in v: v_j\geq \med(v)\}|\geq (n+1)/2$.
We have the following useful lemma for a specific recursive form.

\begin{lemma}\label{lem:ind4}
    Fix $n$. For all $x_0, \cdots, x_n\in A$, and a non-decreasing sequence $v_1, \cdots, v_{n+2}$, there is 
\begin{equation}
    \med(x_1, \cdots, x_n, \med(x_0, v_1, v_2), \cdots, \med(x_0, v_{n+1}, v_{n+2})) 
    = 
    \med(x_0, x_1, \cdots, x_n, v_1, \cdots, v_{n+2}). \label{eq:ind108}
\end{equation}
\end{lemma}

\begin{proof}
    Note that this lemma is exactly a wrap-up of step 4 of the proof of \cite{moulin1980}'s Proposition 2.

    Case 1: $x_0 \leq v_1$. There is $\med(x_0, v_k, v_{k+1}) = v_k$ for all $k = 1, \cdots, n+1$. Thus the left-hand side (LHS) of (\ref{eq:ind108}) becomes $\med(x_1, \cdots, x_n, v_1, \cdots, v_{n+1})$ denoted by $\theta$. Note that $v_1, \cdots, v_{n+1}$ is a majority ($n+1$ terms). It implies that $v_1\leq \theta\leq v_{n+1}$. Thus $\theta$ is still a median if we add $x_0\leq v_1$ and $v_{n+2}\geq v_{n+1}$. That is, $\med(x_0, \cdots, x_n, v_1, \cdots, v_{n+2}) = \theta$.

    Case 2: $x_0 \geq v_{n+2}$. Then $\med(x_0, v_k, v_{k+1}) = v_{k+1}$ for all $k = 1, \cdots, n+1$. It implies that the LHS of (\ref{eq:ind108}) equals $\med(x_1, \cdots, x_n, v_2, \cdots, v_{n+2})$ denoted by $\theta'$. Since $v_2, \cdots, v_{n+2}$ is the majority, we have $v_2\leq \theta'\leq v_{n+2}$. It implies that $\theta'$ is still a median if we add $v_1\leq v_2$ and $x_0\geq v_{n+2}$. It results in $\med(x_0, \cdots, x_n, v_1, \cdots, v_{n+2}) = \theta'$.

    Case 3: $v_k\leq x_0\leq v_{k+1}$ for some $k = 1, \cdots, n+1$. It implies that $\med(x_0, v_{k'}, v_{k'+1}) = v_{k'+1}$ for all $k' = 1, \cdots, k-1$; $\med(x_0, v_k, v_{k+1}) = x_0$; and $\med(x_0, v_{k''}, v_{k''+1}) = v_{k''}$ for all $k'' = k+1, \cdots, n+1$. Then the LHS of (\ref{eq:ind108}) becomes $\med(x_1, \cdots, x_n, v_2, \cdots, v_k, x_0, v_{k+1}, \cdots, v_{n+1})$ denoted by $\theta''$. Since $v_2, \cdots, v_k, x_0, v_{k+1}, \cdots, v_{n+1}$ is the majority, we know $v_1\leq \min\{x_0,v_2\}\leq \theta''\leq \max\{x_0,v_{n+1}\} \leq v_{n+2}$. Then adding $v_1$ and $v_{n+2}$ will not change the median. It gives $\med(x_0, \cdots, x_n, v_1, \cdots, v_{n+2}) = \theta''$, as desired.
\end{proof}

For intervals, consider a vector of intervals $V=(V_1, \cdots, V_n) \in\calI^n$. We can separately arrange their leftmost and rightmost alternatives as $l_1\leq \cdots \leq l_n$ and $r_1\leq \cdots\leq r_n$. 
The rank-based interval operator with parameters $\gamma$, denoted by $\rankI_\gamma(\cdot)$, selects an interval with the $\gamma$-th smallest leftmost and rightmost alternatives. That is, $\rankI_\gamma(V) = [l_\gamma, r_\gamma]$. 


Let $\leftI(\cdot)$ and $\rightI(\cdot)$ denote the leftest and rightest interval operators, respectively. The leftest (rightest) interval operator chooses the leftest (rightest) possible interval, that is, $\leftI(V) = [\min_i V_i^-, \min_i V_i^+] \equiv [l_1, r_1]$ and $\rightI(V) = [\max_i V_i^-, \max_i V_i^+] \equiv [l_n, r_n]$. The leftest and rightest interval operators satisfy the distributivity law.

\begin{lemma}\label{lem:ind6}
    Leftest and rightest interval operators satisfy the following distributivity laws: For all $X,Y,Z\in\calI$, 
\begin{enumerate}[label=(\alph*)]
    \item $\rightI\{X, \leftI\{Y,Z\}\} = \leftI\{\rightI\{X,Y\}, \rightI\{X,Z\}\}$;
    \item $\leftI\{X, \rightI\{Y,Z\}\} = \rightI\{\leftI\{X,Y\}, \leftI\{X,Z\}\}$.
\end{enumerate}
\end{lemma}

\begin{proof}
    We apply the distributivity law of maximum over minimum, i.e., $\max \{a, \min\{b, c\}\} = \min\{\max\{a, b\}, \max\{a, c\}\}$ for all $a,b,c\in A$, to show the statement (a). Note that the outcomes of leftest and rightest interval operators are always feasible. We separately consider leftmost and rightmost alternatives. The leftmost alternative of the LHS of (a) is as follows:
\begin{equation*}
\begin{split}
    {\rightI}^-\{X, \leftI\{Y,Z\}\} 
    &= 
    \max\{X^-, {\leftI}^-\{Y, Z\}\} 
    =
    \max\{X^-, \min\{Y^-, Z^-\}\} 
    =
    \min\{\max\{X^-, Y^-\}, \max\{X^-, Z^-\}\} \\
    &=
    \min\{{\rightI}^-\{X, Y\}, {\rightI}^-\{X, Z\}\} 
    =
    {\leftI}^-\{\rightI\{X, Y\}, \rightI\{X, Z\}\}.
\end{split}
\end{equation*}
    Similarly, the equality for the rightmost alternative of (a) can be obtained. This proves the statement (a). The argument for the statement (b) is analogous and is therefore omitted.
\end{proof}

Suppose $n$ is an odd number. The median interval operator, denoted by $\medI(\cdot)$, is defined as $\medI(V) \equiv [\medI^-(V), \medI^+(V)]$, where $\medI^\pm(V) = \med(V_1^\pm, \cdots, V_n^\pm)$. It selects an interval of which the leftmost and rightmost alternatives are medians of leftmost and rightmost alternatives of peaks. 

Note that the feasibility of these operators is guaranteed by the following lemma. 

\begin{lemma}\label{lem:ind5}
    Fix $n$. Consider a sequence of intervals $V_1, \cdots, V_n\in\calI$, and separately arrange their leftmost and rightmost alternatives as $l_1 \leq \cdots \leq l_n$ and $r_1 \leq \cdots\leq r_n$. Then we have $l_k\leq r_k$ for all $k$. That is, the interval $[l_k, r_k]$ is always feasible.
\end{lemma}

\begin{proof}
    For all $k=1, \cdots, n$, we know that $J = \{j \in\{1, \cdots, n\}: V_j^+\leq r_k\}$ has a cardinality $|J|\geq k$. Note that for all $j\in J$, we have $V_j^-\leq V_j^+\leq r_k$. So that we can always find at least $k$ leftmost alternatives of $V_j$ for all $j\in J$ such that $V_j^- \leq r_k$. It implies $l_k \leq \max_{j\in J} V_j^- \leq r_k$, as desired.
\end{proof}

In particular, for the median interval operator, we can select a feasible interval with endpoints among fixed ballots.

\begin{lemma}\label{lem:ind3}
    Fix $n$. Consider a sequence of intervals $V_1\leq\cdots\leq V_{n+1}\in\calI$. For all $k_1, k_2\in[n+1]$ such that $k_1\leq k_2$, there is
\begin{equation*}
    \medI(\overbrace{\{A^-\}, \cdots, \{A^-\}}^{n-k_2+1}, \overbrace{A, \cdots, A}^{k_2 - k_1}, \overbrace{\{A^+\}, \cdots, \{A^+\}}^{k_1-1}, V_1, \cdots, V_{n+1}) 
    =
    [V_{k_1}^-, V_{k_2}^+].
\end{equation*}
\end{lemma}

\begin{proof}
    Note that $\med(\overbrace{A^-, \cdots, A^-}^{n-k+1}, \overbrace{A^+, \cdots, A^+}^{k-1}, v_1, \cdots, v_{n+1}) = v_k$ for each ordered sequence $v_1\leq \cdots \leq v_{n+1}$. Separately applying it over leftmost alternatives and rightmost alternatives of sequence $V_1, \cdots, V_{n+1}$ with parameters $k_1$ and $k_2$, we obtain the result as desired, where $\{A^-\} = [A^-, A^-]$, $\{A^+\} = [A^+, A^+]$, and $A = [A^-, A^+]$. 
\end{proof}

\subsection{Voting schemes}

We focus on the social choice functions satisfying peak-onlyness. That is, the social choice based on reported preferences depends only on their peaks. We refer to the peak-only social choice function as the voting scheme, and simplify its inputs as peak intervals instead of preferences. Fix $n$. The scheme is a mapping $(X_1, \cdots, X_n)\mapsto C(X_1, \cdots, X_n)$ from the peak profile of $n$ agents to a choice of interval. 

Two important classes of voting schemes are as follows. Fix the number of agents $n$. 
The class of generalized median scheme is defined as
\begin{equation*}
    S_n
    = 
    \{
        C: \calI^n\to \calI
        \mid
        \exists D_1, \cdots, D_{n+1} \in\calI: 
        C(X_1, \cdots, X_n) = \medI (X_1, \cdots, X_n, D_1, \cdots, D_{n+1})
    \},
\end{equation*}
where $D_1, \cdots, D_{n+1}$ are referred to as the fixed ballots. 
The class of leftright scheme for $n$ agents is defined as
\begin{equation*}
    \Sigma_n
    =
    \left\{
        C:\calI^n\to \calI
        \ \middle|\ 
        \exists D_S \in\calI\ \text{for all $S\subseteq N$}: 
        \leftI_{S\subseteq N} \left\{\rightI_{i\in S}\{X_i, D_S\} \right\}
    \right\},
\end{equation*}
where $(D_S)_{S\subseteq N}$ are fixed ballots for all coalitions. 
We provide several examples to illustrate.

\begin{example}
    The leftright scheme and its decomposition.

    We illustrate the scheme for 3 agents. Let $X_1, X_2, X_3\in\calI$ denote agents' peaks. The leftright scheme $\mathrm{LR}\in\Sigma_3$ is
\begin{equation*}
    \mathrm{LR}(X_1, X_2, X_3)
    =
    \leftI\left\{\begin{array}{l}
        D_\varnothing, \rightI \{X_1, D_{\{1\}}\}, \rightI \{X_2, D_{\{2\}}\}, \rightI \{X_3, D_{\{3\}}\},  \\
        \rightI \{X_1, X_2, D_{\{1, 2\}}\}, \rightI \{X_1, X_3, D_{\{1, 3\}}\}, \rightI \{X_2, X_3, D_{\{2, 3\}}\}, \\
        \rightI \{X_1, X_2, X_3, D_{\{1, 2, 3\}}\}
    \end{array}
    \right\},
\end{equation*}
    where $(D_S)_{S\subseteq \{1, 2, 3\}}\in\calI^8$ denote fixed ballots for all coalitions.
    
    Fix $N$. We can also write the leftright scheme separately for leftmost and rightmost alternatives. It consists of two minmax schemes for alternatives instead of intervals, denoted by $\mathrm{mm}_1, \mathrm{mm}_2: A^N\to A$, under the feasibility constraint. Let $X\in\calI^N$ denote agents' peak profile. The scheme $\mathrm{LR}\in\Sigma_n$ is
\begin{equation*}
    \mathrm{LR}(X)
    =
    [\mathrm{mm}_1((X_i^-)_{i\in N}), \mathrm{mm}_2((X_i^+)_{i\in N})].
\end{equation*}
    Denote by $d_S^1, d_S^2\in A$, for all $S\subseteq N$, the fixed ballots of two minmax schemes, respectively. The feasibility constraint for the above decomposition is $d_S^1 \leq d_S^2$ for all $S\subseteq N$.
\end{example}

\begin{example}\label{ex:ind2}
    The generalized median scheme is a subclass of the leftright scheme. 

    Fix $N$. Let $(D_S)_{S\subseteq N}$ denote fixed ballots of the leftright scheme $\mathrm{LR}\in\Sigma_n$. Set $D_S = D_{|S|}$ for all $S\subseteq M$. We obtain the generalized median scheme $\mathrm{GM}\in S_n$ with fixed ballots $D_0, D_1, \cdots, D_n\in\calI$.

    Similar to the leftright scheme, we decompose the generalized median scheme for intervals $\mathrm{GM}: \calI^N\to\calI$ into two generalized median schemes for alternatives, denoted by $\mathrm{gm}_1, \mathrm{gm}_2: A^N\to A$, along with the feasibility constraint. Let $X\in\calI^N$ denote agents' peak profile. The scheme is
\begin{equation*}
    \mathrm{GM}(X)
    =
    [\mathrm{gm}_1((X_i^-)_{i\in N}), \mathrm{gm}_2((X_i^+)_{i\in N})].
\end{equation*}
    Let $d_1^1 \leq \cdots \leq d_{n+1}^1\in A$ and $d_1^2 \leq \cdots \leq d_{n+1}^2\in A$ denote the ordered fixed ballots of two generalized median schemes for alternatives, respectively. The feasibility constraint for the above decomposition is $d_j^1 \leq d_j^2$ for all $j = 1, \cdots, n+1$.
\end{example}

In particular, the generalized median scheme and the leftright scheme coincide in the case of $n=1$. 

\begin{lemma}\label{lem:ind2}
    $S_1 = \Sigma_1$. 
\end{lemma}

\begin{proof}
    Fix $D_\varnothing, D_{\{1\}} \in\calI$.
    For all $X_1\in\calI$, we show $\leftI\{D_\varnothing, \rightI\{X_1, D_{\{1\}}\}\} = \medI(X_1, D_\varnothing, D_{\{1\}})$. 

    Without any loss, we assume $D_{\{1\}}\leq D_\varnothing$. Suppose $D_\varnothing^-\leq D_{\{1\}}^-$ for example. We have $D_\varnothing^-\leq\rightI^-\{X_1, D_{\{1\}}\}$, and thus $D_{\{1\}}^-$ plays no role in $\leftI^-\{D_\varnothing, \rightI\{X_1, D_{\{1\}}\}\}$. Then we can set it as $D_{\{1\}}^- = D_\varnothing^-$. A similar argument applies to the rightmost alternative. So $D_{\{1\}}\leq D_\varnothing$ is always kept.

    Let $l_1 = D_{\{1\}}^-$, $r_1 = D_{\{1\}}^+$, $l_2 = D_\varnothing^-$, $r_2 = D_\varnothing^+$, where $l_1\leq l_2$ and $r_1\leq r_2$. 
    Consider the leftmost alternatives.
    If $X_1^-\leq l_1$, then $\rightI^-\{X_1, D_{\{1\}}\} = \max\{X_1^-, l_1\} = l_1$, and thus $\leftI^-\{D_\varnothing, \rightI\{X_1, D_{\{1\}}\}\} = \min\{l_2,l_1\} = l_1$, which coincides with $\medI^-(X_1, D_\varnothing, D_{\{1\}}) = \med(X_1^-, l_1, l_2) = l_1$. If $X_1^-\geq l_2$, we have $\rightI^-\{X_1, D_{\{1\}}\} = \max\{X_1^-, l_1\} = X_1^-$ and $\leftI^-\{D_\varnothing, \rightI\{X_1, D_{\{1\}}\}\} = \min\{l_2, X_1^-\} = l_2=\medI^-(X_1, D_\varnothing, D_{\{1\}})$. If $l_1\leq X_1^-\leq l_2$, there is $\rightI^-\{X_1, D_{\{1\}}\} = \max\{X_1^-, l_1\} = X_1^-$ and $\leftI^-\{D_\varnothing, \rightI\{X_1, D_{\{1\}}\}\} = \min\{l_2, X_1^-\} = X_1^- = \medI^-(X_1, D_\varnothing, D_{\{1\}})$. Similarly, we get $\leftI^+\{D_\varnothing, \rightI\{X_1, D_{\{1\}}\}\} = \medI^+(X_1, D_\varnothing, D_{\{1\}})$ for the rightmost alternatives.
\end{proof}

\begin{example}
    The constant scheme belongs to the generalized median scheme, and thus it belongs to the leftright scheme. 

    Fix $n$. Let $(D_j)_{j = 1}^{n+1}\in\calI^{n+1}$ denote the fixed ballots of the generalized median scheme $\mathrm{GM}\in S_n$. Set $D_j = D_0$ for all $j = 1, \cdots, n+1$. We obtain the constant scheme $\mathrm{CS}(\cdot)\equiv D_0$.
\end{example}

\begin{example}
    The dictatorial scheme belongs to the leftright scheme. 

    Fix $N$. Let $(D_S)_{S\subseteq N}\in\calI^{2^n}$ denote the fixed ballots of the leftright scheme $\mathrm{GM}\in S_n$. Suppose agent $i\in N$ is the dictator. Set $D_{\{i\}} = \{A^-\}$ and $D_S = \{A^+\}$ for all $S\neq \{i\}$. We get $\rightI\{X_i, D_{\{i\}}\} = X_i \leq \{A^+\}$ and $\rightI\{X_S, D_S\} = \{A^+\}$ for all $S\neq\{i\}$. Then $\mathrm{LR}(X) = \leftI\{X_i, \{A^+\}, \cdots, \{A^+\}\} = X_i$, which is a dictatorial scheme. 
\end{example}

\begin{example}
    The leftright scheme allows the mixture. 

    Fix $N$. A voting scheme $C:\calI^N\to\calI$, of which the leftmost alternative follows a generalized median scheme for alternatives $\mathrm{gm}:A^N\to A$ and the rightmost alternative is a constant scheme, also belongs to the leftright scheme under the feasibility constraint. Let $X\in\calI^N$ denote agents' peak profile. The scheme is
\begin{equation*}
    C(X)
    =
    [\mathrm{gm}((X_j^-)_{j\in N}), d_0],
\end{equation*}
    for some $d_0\in A$.
    Let $(d_j)_{j = 1}^{n+1}\in A^{n+1}$ denote fixed ballots of $\mathrm{gm}$. The feasibility constraint is $d_j \leq d_0$ for all $j = 1, \cdots, n+1$.
\end{example}

\begin{example}
    The leftright scheme can be reduced into the minmax scheme for alternatives.

    Fix $N$. Consider the following mixed voting scheme $C:\calI^N\to\calI$. Let $X\in\calI^N$ denote the peak profile. The scheme is
\begin{equation*}
    C(X)
    = 
    [\mathrm{mm}((X_i^-)_{i\in N}), A^+]. 
\end{equation*}
    Note that there is no constraint on fixed ballots of the minmax scheme. In this scheme, only leftmost alternatives of reported peaks matter, which follows the one-dimensional minmax scheme as in the traditional setting. 
\end{example}

In the class of generalized median scheme, we introduce more specific schemes. Fix $n$.
The restricted generalized median scheme is a generalized median scheme with $n-1$, instead of $n+1$, fixed ballots. It can be represented by a generalized median scheme such that two fixed ballots are set as $\{A^-\}$ and $\{A^+\}$, respectively. 



The rank-based scheme selects the $\gamma$-th leftest interval, i.e., $\rankI_\gamma(X)$, constructed by the peak profile $X\in\calI^N$ for some $\gamma$. It can be represented as a restricted generalized median scheme such that $\gamma$ fixed ballots are $\{A^+\}$ and other fixed ballots are $\{A^-\}$. 
Fix an odd $n$, the median scheme $\medI(X)\equiv\rankI_{\frac{n+1}{2}}(X)$ is a rank-based scheme with parameter $(n+1)/2$. 

\begin{example}
    A numerical example for the rank-based scheme and the median scheme.

    Consider the peak profile $X = ([1, 8], [3, 4], [2, 3], [7, 8], [2, 2])$ . We separately sort their leftmost and rightmost alternatives in ascending order as (1, 2, 2, 3, 7) and (2, 3, 4, 8, 8). Then the 1-th rank-based scheme selects $\rankI_1(X)=[1,2]$; the 4-th rank-based scheme selects $\rankI_4(X)=[3, 8]$; and the median scheme, as well as the 3-th rank-based scheme, selects $\medI(X) = \rankI_3(X) = [2, 4]$.
\end{example}

\section{Strategy-proof voting scheme}\label{sec:ind3}

\subsection{Strategy-proofness, uncompromisingness, and separability}

In this section, we provide some preliminaries for the analysis based on strategy-proofness. Formally, the axiom is defined as follows: 

\begin{axiom}[Strategy-proofness]
    Fix $N$. For all $\succsim\in\mathcal{P}^N$ with peaks $X\in\calI^N$, all $i\in N$, and all $Y_i\in \calI$, there is $C(X) \succsim_i C(Y_i, X_{-i})$.
\end{axiom}

Strategy-proofness says that reporting the true peak is a weakly dominant strategy for each agent. 

This axiom can be translated into a geometric property of the choice function. That is, the original choice should be located between the peak of the agent who misreports and the manipulated choice. This approach is used in \cite{moulin1980} and then generalized and formalized into uncompromisingness in \cite{border1983}.

\begin{axiom}[Uncompromisingness]
    Fix $N$. For all $X\in\calI^N$, all $i\in N$, and all $Y_i\in \calI$, there is $C(X) \in B(X_i, C(Y_i, X_{-i}))$.
\end{axiom}

\begin{lemma}\label{lem:ind1}
    Strategy-proofness is equivalent to uncompromisingness.
\end{lemma}

\begin{proof}
    Fix $N$. Suppose a voting scheme $C:\calI^N\to\calI$ satisfies strategy-proofness, but for some $X\in\calI^N$, $i\in N$, and $Y_i\in\calI$, there is $C(X)\notin B(X_i, C(Y_i, X_{-i}))$. By richness condition, there exists a preference $\succsim_i\in\mathcal{P}$ such that $C(Y_i, X_{-i})\succ_i C(X)$. Note that strategy-proofness requires $C(X) \succsim_i C(Y_i, X_{-i})$ for all $\succsim_i\in\mathcal{P}$. This leads to a contradiction. Therefore, we must have $C(X)\in B(X_i, C(Y_i, X_{-i}))$. 

    Conversely, suppose the voting scheme $C$ satisfies uncompromisingness. Consider any peak profile $X\in\calI^N$, agent $i\in N$, and a deviation $Y_i\in\calI$. From uncompromisingness, we have $C(X)\in B(X_i, C(Y_i, X_{-i}))$. 
    If $X_i = C(Y_i, X_{-i})$, we have $X_i = C(X) = C(Y_i, X_{-i})$, and thus $C(X)\succsim_i C(Y_i, X_{-i})$ always holds, as desired.
    Suppose $X_i \neq C(Y_i, X_{-i})$. Since the agent $i$ has a single-peaked preference, it implies $X_i\succsim_i C(X)\succ_i C(Y_i, X_{-i})$ if $C(X)\neq C(Y_i, X_{-i})$. Note that $C(X)\succsim_i C(Y_i, X_{-i})$ always holds if $C(X)= C(Y_i, X_{-i})$. The strategy-proofness is thus verified.
\end{proof}

Another consequence of strategy-proofness on a multidimensional single-peaked domain is separability \citep*{barbera1993}. 

\begin{axiom}[Separability]
    Fix $N$. There exist two functions $f_l, f_r: A^N\to A$ such that for all $X\in\calI^N$, $C^-(X) = f_l((X_i^-)_{i\in N})$ and $C^+(X) = f_r((X_i^+)_{i\in N})$.
\end{axiom}

\begin{lemma}\label{lem:ind10}
    Strategy-proofness implies separability.
\end{lemma}

\begin{proof}
    Fix $N$. Suppose a voting scheme $C:\calI^N\to \calI$ satisfies strategy-proofness, and thus uncompromisingness. We show that there exists a left function $f_l: A^N\to A$ such that for all $X\in\calI^N$, $C^-(X) = f_l((X_i^-)_{i\in N})$. That is, the leftmost alternative of the social choice only depends on the leftmost alternatives of intervals $X$. The claim for the right function $f_r$ follows the same argument.

    Consider any peak profile $X\in \calI^N$, an agent $i\in N$, and an interval $Y_i\in \calI$ such that $X_i^- = Y_i^-$. From uncompromisingness, we have $C(X) \in B(X_i, C(Y_i, X_{-i}))$ and $C(Y_i, X_{-i}) \in B(Y_i, C(X))$. It implies that 
\begin{align}
    C^-(X)
    &\in
    [\min\{X_i^-, C^-(Y_i, X_{-i})\}, \max\{X_i^-, C^-(Y_i, X_{-i})\}]; \label{eq:ind109a}
    \\
    C^-(Y_i, X_{-i})
    &\in
    [\min\{Y_i^-, C^-(X)\}, \max \{Y_i^-, C^-(X)\}] 
    = 
    [\min\{X_i^-, C^-(X)\}, \max \{X_i^-, C^-(X)\}]. \label{eq:ind109b}
\end{align}

    We show that $C^-(X) = C^-(Y_i, X_{-i})$. 
    Suppose $C^-(X) < C^-(Y_i, X_{-i})$. (\ref{eq:ind109a}) is equivalent to $X_i^-\leq C^-(X) < C^-(Y_i, X_{-i})$; and (\ref{eq:ind109b}) is equivalent to $C^-(X) < C^-(Y_i, X_{-i})\leq X_i^-$. Then $X_i^-\leq C^-(X) < C^-(Y_i, X_{-i})\leq X_i^-$ leads to a contradiction. Thus $C^-(X) \geq C^-(Y_i, X_{-i})$. 
    Similarly, suppose $C^-(X) > C^-(Y_i, X_{-i})$. (\ref{eq:ind109a}) is equivalent to $C^-(Y_i, X_{-i}) < C^-(X) \leq X_i^-$, while (\ref{eq:ind109b}) implies that $C^-(Y_i, X_{-i})\geq X_i^-$. This constructs a contradiction. So $C^-(X) \leq C^-(Y_i, X_{-i})$. 
    
    From $X_i^- = Y_i^-$ and $C^-(X) = C^-(Y_i, X_{-i})$, we know that $C^-(X)$ is independent of the choice of rightmost alternative $X_i^+$. Note that $i$ is arbitrarily chosen. It implies that $C^-(\cdot)$ only depends on the leftmost alternatives, as desired.  
\end{proof}

We provide an example to show how two properties, i.e., uncompromisingness and separability, apply to the voting scheme. 

\begin{lemma}\label{lem:ind11}
    Strategy-proofness implies that a voting scheme for one agent belongs to $S_1\equiv\Sigma_1$. 
\end{lemma}

\begin{proof}
    From Lemma~\ref{lem:ind2}, we have shown $S_1 = \Sigma_1$. We only show that strategy-proofness implies the voting scheme $C:\calI\to\calI$ belongs to $S_1$. That is, for all $X_1\in\calI$, we have $C(X_1) = \medI(X_1, D_1, D_2)$ for some $D_1, D_2\in\calI$. 

    We construct two fixed ballots $D_1, D_2$. Let $l_1 = \min_{V\in\calI} C^-(V)$, $l_2 = \max_{V\in\calI} C^-(V)$, $r_1 = \min_{V\in\calI} C^+(V)$, and $r_2 = \max_{V\in\calI} C^+(V)$. Note that $l_1\leq r_1$ and $l_2\leq r_2$. Set $D_1 = [l_1, r_1]$ and $D_2 = [l_2, r_2]$. Suppose $L_1, R_1$ are intervals such that $C^-(L_1) = l_1$ and $C^+(R_1) = r_1$. From separability, we have $C^-(L_1) = C^-(D_1) = l_1$, and $C^+(R_1) = C^+(D_1) = r_1$. Similarly, we obtain $C(D_2) = [l_2, r_2]$.
    
    We consider the leftmost alternative. From uncompromisingness, we have $C(X_1)\in B(X_1, C(D_1))$ and $C(X_1)\in B(X_1, C(D_2))$. Then we have
\begin{align}
    \min \{X_1^-, l_1\} \leq C^-(X_1) \leq \max \{X_1^-, l_1\} \label{eq:ind101a}\\
    \min \{X_1^-, l_2\} \leq C^-(X_1) \leq \max \{X_1^-, l_2\} \label{eq:ind101b}
\end{align}
    If $X_1^- \leq l_1$, (\ref{eq:ind101a}) implies $C^-(X_1)\leq l_1 = \min_{V\in\calI} C^-(V)$, and thus $C^-(X_1) = l_1$. 
    Similarly, if $X_1^- \geq l_2$, (\ref{eq:ind101b}) implies $\max_{V\in\calI} C^-(V) = l_2 \leq C^-(X_1)$, and thus $C^-(X_1) = l_2$. If $l_1\leq X_1^-\leq l_2$, (\ref{eq:ind101a}) implies $C^-(X_1)\leq X_1^-$, and (\ref{eq:ind101b}) implies $X_1^- \leq C^-(X_1)$. Thus we get $X_1^- = C^-(X_1)$. 
    
    Notice that the above analysis shows that $C^-(X_1) = \med (X_1^-, D_1^-, D_2^-)$. A similar argument shows that $C^+(X_1) = \med (X_1^+, D_1^+, D_2^+)$. That is, $C(X_1) = \medI(X_1, D_1, D_2)$, as desired.
\end{proof}

\subsection{The leftright scheme}

\begin{theorem}\label{thm:ind5}
    The leftright scheme is the only strategy-proof voting scheme.
\end{theorem}

\begin{proof}
    Fix $n$. To assist our analysis, we assume that in the class of leftright schemes $\Sigma_n$, there is $D_S \geq D_T$ for two coalitions $S, T\subseteq N$ such that $S\subseteq T$. Suppose, for example, $D_S^- \leq D_T^-$. Combining $\rightI_{i\in S}^- X_i \leq \rightI_{i\in T}^- X_i$, we get $\rightI_{i\in S}^-\{X_i, D_S\}\leq \rightI_{i\in T}^-\{X_i, D_T\}$. Thus in $\leftI_{S'\subseteq N} \{\rightI_{i\in S'}\{X_i, D_{S'}\}\}$, $D_T^-$ plays no role, and then it can be set as $D_T^- = D_S^-$ in this case. Similarly, we can set $D_T^+ = D_S^+$ in the case of $D_S^+ \leq D_T^+$. 
    
    We show that the leftright scheme satisfies strategy-proofness. Fix $N$. Consider a profile $X\in\calI^N$, an agent $i\in N$, and a deviation $Y_i\in\calI$. We want to show uncompromisingness, i.e., $C(X)\in B(X_i, C(Y_i, X_{-i}))$. We show the claim for the leftmost alternative.
    
    Note that $C^-(X) = \leftI_{S\subseteq N}^-\{ \rightI_{j\in S}\{X_j, D_S\} \} = \min_{S\subseteq N} \{\max_{j\in S}\{X_j^-, D_S^-\}\}$. 

    Case 1: $C^-(X) = X_i^-$. Then, for all $C(Y_i, X_{-i})$, we have either $C^-(Y_i, X_{-i})\leq C^-(X) = X_i^-$ or $C^-(Y_i, X_{-i})\geq C^-(X) = X_i^-$.  

    Case 2: $C^-(X) > X_i^-$. We need to show $C^-(X)\leq C^-(Y_i, X_{-i})$ for all $Y_i\in\calI$. 
    If $Y_i^-\geq X_i^-$, then for all $S\subseteq N$ such that $i\in S$, there is $\max_{j\in S\setminus\{i\}}\{Y_i^-, X_j^-, D_S^-\} \geq \max_{j\in S}\{X_j^-, D_S^-\}$, which implies $C^-(Y_i, X_{-i})\geq C^-(X)$. 
    If $Y_i^-< X_i^-$, suppose $C^-(X)> C^-(Y_i, X_{-i})$. That is, for some $S^*$ such that $i\in S^*$, the deviation from $X_i$ to $Y_i$ decreases the term $\max_{j\in S^*}\{X_j, D_{S^*}\}$ and $C^-(X)$, and gives that $C^-(Y_i, X_{-i}) = \max_{j\in S^*\setminus\{i\}}\{Y_i, X_j, D_{S^*}\}$. This requires that $X_i^-$ is responsible for the maximum term. That is, $X_i^- = \max_{j\in S^*}\{X_j, D_{S^*}\}\geq C^-(X)$, where the second inequality is from the definition of $C^-(\cdot)$. This leads to a contradiction.

    Case 3: $C^-(X) < X_i^-$. We show that $C^-(X)\geq C^-(Y_i, X_{-i})$ for all $Y_i\in\calI$. If $Y_i^-\leq X_i^-$, each maximum term cannot be increased and thus $C^-(X)\geq C^-(Y_i, X_{-i})$ is obtained. If $Y_i^- > X_i^-$, consider the pivotal coalition $S^*$ such that $C^-(X) = \max_{j\in S^*}\{X_j, D_{S^*}\}$. Note that $i\notin S^*$. If so, i.e., $i\in S^*$, we get $C^-(X) = \max_{j\in S^*}\{X_j, D_{S^*}\}\geq X_i$, contradicting $C^-(X) < X_i^-$. Thus, the deviation from $X_i$ to $Y_i$ will not increase $C^-(X)$. Thus it proves $C^-(X)\geq C^-(Y_i, X_{-i})$.

    The claim for the rightmost alternative follows a symmetric argument. We conclude that the leftright scheme satisfies uncompromisingness and thus strategy-proofness.

    We argue the converse claim by induction. Statement $n$ is that every voting scheme $C:\calI^n\to\calI$ among $n$ agents satisfying strategy-proofness belongs to $\Sigma_n$. 

    The statement 1 is implied by Lemma~\ref{lem:ind11}.
    Assume statements $1, \cdots, n$ hold. We show the statement $n+1$ below.

    Let $C(X_0, X_1, \cdots, X_n)$ be the scheme among $n+1$ agents satisfying strategy-proofness. Fix $X_0$. From the statement $n$, we have $(X_1, \cdots, X_n)\mapsto C(X_0, X_1, \cdots, X_n)$ belongs to $\Sigma_n$. Note that here we abuse $N = \{1, \cdots, n\}$. Then there exists a collection of intervals $(D_S)_{S\subseteq N}$, satisfying $S\subseteq T$ implies $D_S\geq D_T$ for all $S,T\subseteq N$, such that
\begin{equation}
    C(X_0, X_1, \cdots, X_n)
    =
    \leftI_{S\subseteq N} 
    \left\{
        \rightI_{i\in S} \{X_i, D_S(X_0)\}
    \right\}.\label{eq:ind110}
\end{equation}

    Consider a non-empty coalition $S\subseteq N$. Let $X_i = \{A^-\}$ for all $i\in S$; and $X_i = \{A^+\}$ for all $i\notin S$. We have $\rightI_{i\in T} \{X_i, D_T(X_0)\} = D_T(X_0)$ for all $T\subseteq S$; and $\rightI_{i\in T} \{X_i, D_T(X_0)\} = \{A^+\}$, which doesn't affact the result of leftest interval operator, for all $T\subseteq N$ such that $T\setminus S\neq \varnothing$. Therefore (\ref{eq:ind110}) becomes
\begin{equation*}
    C(X_0, X_1, \cdots, X_n)
    =
    \leftI_{T\subseteq S} 
    \left\{
        \rightI_{i\in T} \{X_i, D_T(X_0)\}
    \right\}
    =
    \leftI_{T\subseteq S} \{D_T(X_0)\}
    =
    D_S(X_0).
\end{equation*}
    Similarly, we replace $X_0$ by $Y_0$, and get $C(Y_0, X_1, \cdots, X_n) = D_S(Y_0)$.

    Now fix $(X_1, \cdots, X_n)$. We know $X_0 \mapsto C(X_0, X_1, \cdots, X_n)$ satisfies strategy-proofness and thus uncompromisingness. It implies $C(X_0, X_1, \cdots, X_n)\in B(X_0, C(Y_0, X_1, \cdots, X_n))$, which is equivalent to $D_S(X_0)\in B(X_0, D_S(Y_0))$. That is, $D_S(\cdot)$ satisfies uncompromisingness, and thus belongs to $\Sigma_1$. Therefore we can write it as
\begin{equation}
    D_S(X_0) 
    = 
    \leftI\{
        E_{S,2}, \rightI\{X_0, E_{S,1}\}
    \}, \label{eq:ind111}
\end{equation}
    where $E_{S,1}\leq E_{S,2}$.

    Finally we derive the leftright scheme for agents $N\cup\{0\}$.
    Substituting (\ref{eq:ind111}) into (\ref{eq:ind110}), we have
\begin{equation}
\begin{split}
    C(X_0, X_1, \cdots, X_n)
    &=
    \leftI_{S\subseteq N} 
    \left\{
        \rightI_{i\in S} 
        \{
            X_i, \leftI\{ E_{S,2}, \rightI\{X_0, E_{S,1}\} \}
        \}
    \right\} 
    =
    \leftI_{S\subseteq N} 
    \left\{
        \leftI
        \left\{
            \rightI_{i\in S}\{X_i, E_{S,2}\}, \rightI_{i\in S}\{ X_i, \rightI\{X_0, E_{S,1}\} \}
        \right\}
    \right\} \\
    &=
    \leftI_{S\subseteq N} 
    \left\{
        \leftI
        \left\{
            \rightI_{i\in S}\{X_i, E_{S,2}\}, \rightI_{i\in S}\{ X_i, X_0, E_{S,1} \}
        \right\}
    \right\} 
    = 
    \leftI_{S\subseteq N} 
    \left\{
        \leftI
        \left\{
            \rightI_{i\in S}\{X_i, E_{S,2}\}, \rightI_{i\in S\cup\{0\}}\{X_i, E_{S,1}\}
        \right\}
    \right\},
    \label{eq:ind112}
\end{split}
\end{equation}
    where the second step is from the distributivity law of rightest operator over leftest operator by Lemma~\ref{lem:ind6}. 
    Let $F_S = E_{S,2}$ and $F_{S\cup\{0\}} = E_{S,1}$ for all $S\subseteq N$. We write (\ref{eq:ind112}) as
\begin{equation*}
\begin{split}
    C(X_0, X_1, \cdots, X_n)
    &=
    \leftI_{S\subseteq N} 
    \left\{
        \leftI
        \left\{
            \rightI_{i\in S}\{X_i, F_S\}, \rightI_{i\in S\cup\{0\}}\{X_i, F_{S\cup\{0\}}\}
        \right\}
    \right\}
    =
    \leftI_{S\subseteq N\cup\{0\}}
    \left\{
        \rightI_{i\in S}\{X_i, F_S\}
    \right\},
\end{split}
\end{equation*}
    as desired.
\end{proof}

\section{The (generalized) median schemes}\label{sec:ind4}

\subsection{Anonymity and the generalized median scheme}
Next, we introduce anonymity. We call $\pi: N\to N$ a permutation of $N$ if it is a bijection; and denote the permuted peak profile of agents $N$ by $X_\pi = (X_{\pi(i)})_{i\in N}$. 

\begin{axiom}[Anonymity]
    Fix $N$. For all $X\in\calI^N$ and any permutation $\pi$ of $N$, there is $C(X) = C(X_\pi)$.
\end{axiom}

This axiom says that the social choice is irrelevant to agents' identities. Anonymity is particularly important in voting mechanisms as it reveals the principle of ``one person, one vote.'' 

Adding anonymity, the leftright scheme reduces to the generalized median scheme. As shown in Example~\ref{ex:ind2}, anonymity requires the fixed ballots for all coalitions with the same cardinality to be the same, i.e., $D_S = D_{|S|}$ for all $S\subseteq N$. Then fixed ballots for coalitions $(D_S)_{S\subseteq N}$ is reduced to $n+1$ intervals $D_0, \cdots, D_n$. The generalized median scheme selects the median interval among peak profile and $n+1$ fixed ballots. 

The following theorem is also presented in \cite{miller2025}. We provide an alternative proof by induction.

\begin{theorem}\label{thm:ind1}
    The generalized median scheme is the only anonymous and strategy-proof voting scheme.
\end{theorem}

\begin{proof}
    Since median interval operator is independent of agents' identities, anonymity is obvious. We show the strategy-proofness. 

    Fix $N$. Consider peaks $X\in\calI^N$ and fixed ballots $D_1, \cdots, D_{n+1}\in\calI$. For all $i\in N$, fix $X_{-i}$ and $D_1, \cdots, D_{n+1}$, and consider a deviation $Y_i\in\calI$. Note that $C$ is a generalized median rule, and the median interval operator aggregates leftmost and rightmost alternatives of $X$ and $D_1, \cdots, D_{n+1}$ separately. Thus in our analysis, the intervals can be decomposed into two ordered sequences. We arrange the leftmost alternatives of $X_{-i}$ and $D_1, \cdots, D_{n+1}$ as $l_1 \leq \cdots \leq l_{2n}$ and their rightmost alternatives as $r_1 \leq \cdots \leq r_{2n}$. 
    
    We show that $C^-(X)$ lies between $X_i^-$ and $C^-(Y_i, X_{-i})$.
    Note that
\begin{equation*}
    C^-(X)
    =
    \left\{\begin{array}{ll}
        l_n, & \text{if $X_i^-\leq l_n$}, \\
        X_i^-, & \text{if $l_n\leq X_i^-\leq l_{n+1}$}, \\
        l_{n+1}, & \text{if $X_i^-\geq l_{n+1}$}.
    \end{array}
    \right.
    =
    \med(X_i^-, l_n, l_{n+1}).
\end{equation*}
    Similarly, we have $C^-(Y_i, X_{-i}) = \med(Y_i^-, l_n, l_{n+1})$.
    
    Case 1: $X_i^-\leq l_n$. We have $C^-(X)=\med(X_i^-, l_n, l_{n+1}) = l_n\geq X_i^-$. Also, for all $Y_i^-\in A$, there is $C^-(Y_i, X_{-i}) = \med(Y_i^-, l_n, l_{n+1}) \geq l_n = C^-(X)$, as desired.

    Case 2: $X_i^-\geq l_{n+1}$. There is $C^-(X) = \med(X_i^-, l_n, l_{n+1}) = l_{n+1}\leq X_i^-$. For all $Y_i^-\in A$, we have $C^-(Y_i, X_{-i}) = \med(Y_i^-, l_n, l_{n+1}) \leq l_{n+1} = C^-(X)$, as desired.

    Case 3: $l_n\leq X_i^-\leq l_{n+1}$. We have $C^-(X)=\med(X_i^-, l_n, l_{n+1}) = X_i^-$, which satisfies the desired property regardless of the choice of $C^-(Y_i, X_{-i})$.

    By a similar argument, we get that $C^+(X)$ is between $X_i^+$ and $C^+(Y_i, X_{-i})$. Uncompromisingness is therefore satisfied. This establishes the strategy-proofness.
    
    Conversely, suppose the scheme satisfies the strategy-proofness and anonymity. 
    We argue that it is a generalized median scheme by induction. 
    
    The statement $n$ is that every voting scheme $C:\calI^n\to\calI$ among $n$ agents satisfying strategy-proofness and anonymity belongs to $S_n$. 
    
    The statement 1 is implied by Lemma~\ref{lem:ind11}. Assume statements $1, \cdots, n$ hold. We show the statement $n+1$.
    
    Let $C(X_0, X_1, \cdots, X_n)$ denote a voting scheme among $n+1$ agents satisfying strategy-proofness and anonymity. Fix $X_0$. We know that $(X_1, \cdots, X_n) \mapsto C(X_0; X_1, \cdots, X_n)$ also satisfies strategy-proofness and anonymity, and thus belongs to $S_n$ from the statement $n$. It says that there exist $n+1$ functions $E_1, \cdots, E_{n+1}: \calI\to\calI$ such that for all $(X_1, \cdots, X_n)\in\calI^n$, and all $X_0\in\calI$,
    $C(X_0, X_1, \cdots, X_n) = \medI(X_1, \cdots, X_n, E_1(X_0), \cdots, E_{n+1}(X_0))$

    We organize the leftmost and rightmost alternatives of the sequence $E_1(X_0), \cdots, E_{n+1}(X_0)$, and construct a new ordered sequence of fixed ballots. For all $X_0\in\calI$, we organize the sequence of leftmost alternatives $E_1^-(X_0), \cdots, E_{n+1}^-(X_0)$ as $l_1(X_0) \leq \cdots \leq l_{n+1}(X_0)$, and the sequence of rightmost alternatives $E_1^+(X_0), \cdots, E_{n+1}^+(X_0)$ as $r_1(X_0) \leq \cdots \leq r_{n+1}(X_0)$. Then we construct $F_i(X_0) = [l_i(X_0), r_i(X_0)]$ for all $i=1\cdots, n+1$. Lemma~\ref{lem:ind5} ensures that these intervals are feasible.
    Since median interval operator is invariant to the re-ordering of leftmost and rightmost alternatives, we have
\begin{equation}
    C(X_0,X_1, \cdots, X_n) =\medI(X_1, \cdots, X_n, F_1(X_0), \cdots, F_{n+1}(X_0)), \label{eq:ind102}
\end{equation}
    where $F_1(X_0) \leq \cdots \leq F_{n+1}(X_0)$.

    Fix $(X_1, \cdots, X_n)$. We know $X_0\mapsto C(X_0, X_1, \cdots, X_n)$ satisfies uncompromisingness with respect to $X_0$, and thus belongs to $S_1$ from the statement 1. For all $k=1, \cdots, n+1$, let $X_i = \{A^-\}$ for all $i = 1,\cdots, n-k+1$ if $k\leq n$; and $X_j = \{A^+\}$ for all $j = n-k+2, \cdots, n$ if $k\geq 2$. From Lemma~\ref{lem:ind3}, we get $C(X_0, X_1, \cdots, X_n) = F_k(X_0)$. Since it belongs to $S_1$, we write $F_k(X_0) = \medI(X_0, G_{k,1}, G_{k,2})$, where $G_{k,1}, G_{k,2}\in\calI$ such that $G_{k,1}\leq G_{k,2}$.
    Substituting $F_k(X_0)$ into (\ref{eq:ind102}) for all $k = 1, \cdots, n+1$, we obtain 
\begin{equation}
    C(X_0; X_1, \cdots, X_n) 
    = 
    \medI(X_1, \cdots, X_n, \medI(X_0, G_{1,1}, G_{1,2}), \cdots, \medI(X_0, G_{n+1,1}, G_{n+1,2})), \label{eq:ind104}
\end{equation}
    where $\medI(X_0, G_{1,1}, G_{1,2}) \leq \cdots \leq \medI(X_0, G_{n+1,1}, G_{n+1,2})$.

    We show that $G_{k,1}\leq G_{k+1,1}$ and $G_{k,2}\leq G_{k+1,2}$ for all $k=1, \cdots, n$. Note that for all $X_0\in\calI$, $F_k(X_0)\leq F_{k+1}(X_0)$ implies $\medI (X_0, G_{k,1}, G_{k,2}) \leq \medI (X_0, G_{k+1,1}, G_{k+1,2})$. Let $X_0 = \{A^-\}$.  
    Combining $G_{k,1}\leq G_{k,2}$ and $G_{k+1,1}\leq G_{k+1,2}$, we get $\medI (X_0, G_{k,1}, G_{k,2}) = G_{k,1}$ and $\medI (X_0, G_{k+1,1}, G_{k+1,2}) = G_{k+1,1}$, and thus $G_{k,1}\leq G_{k+1,1}$. Similarly, setting $X_0 = \{A^+\}$ obtains $G_{k,2}\leq G_{k+1,2}$.

    We claim that $G_{k,2} = G_{k+1,1}$, which is equivalent to $G_{k,2}^- = G_{k+1,1}^-$ and $G_{k,2}^+ = G_{k+1,1}^+$, for all $k=1, \cdots, n$. We show the left claim $G_{k,2}^- = G_{k+1,1}^-$. Note that the leftmost alternative of (\ref{eq:ind104}) can be represented as
\begin{equation}
\begin{split}
    C^-(X_0, X_1, \cdots, X_n) 
    &= 
    \med(X_1^-, \cdots, X_n^-, F_1^-(X_0), \cdots, F_{n+1}^-(X_0)) \\
    &= 
    \med(X_1^-, \cdots, X_n^-, \med(X_0^-, G_{1,1}^-, G_{1,2}^-), \cdots, \med(X_0^-, G_{n+1,1}^-, G_{n+1,2}^-)).\label{eq:ind105}
\end{split}
\end{equation}

    Suppose $G_{k,2}^- > G_{k+1,1}^-$ for some $k$. Consider a vector of intervals $(X_0, \cdots, X_n)$ such that $X_0 = \{G_{k+1,1}^-\}$, $X_i=\{A^-\}$ for all $i=1,\cdots,n-k$, $X_j= \{A^+\}$ for all $j=n-k+1,\cdots,n-1$, and $X_n = \{G_{k,2}^-\}$. We have $F_k^-(X_0) = \med(G_{k+1,1}^-, G_{k,1}^-, G_{k,2}^-) = G_{k+1,1}^-$ and $F_{k+1}^-(X_0) = \med(G_{k+1,1}^-, G_{k+1,1}^-, G_{k+1,2}^-) = G_{k+1,1}^-$. Thus $F_k^-(X_0) = F_{k+1}^-(X_0) = G_{k+1,1}^- = X_0^-$.
    Combining $F_1(X_0) \leq \cdots \leq F_{n+1}(X_0)$, we have $F_1^-(X_0) \leq \cdots \leq F_k^-(X_0) = X_0^- = F_{k+1}^-(X_0) \leq \cdots \leq F_{n+1}^-(X_0)$. Thus (\ref{eq:ind105}) becomes
\begin{equation*}
    C^-(X_0, X_1, \cdots, X_n) 
    = 
    \med \Bigg(\begin{array}{l}
        \overbrace{A^-, \cdots, A^-}^{n-k}, \overbrace{A^+, \cdots, A^+}^{k-1}, X_n^-, \\
        \underbrace{F_1^-(X_0), \cdots, F_{k-1}^-(X_0)}_{k-1}, X_0^-, X_0^-, \underbrace{F_{k+2}^-(X_0), \cdots, F_{n+1}^-(X_0)}_{n-k}
    \end{array}
    \Bigg) 
    =
    X_0^-.
\end{equation*}
    We consider a permutation only exchanging $0$ and $n$. There is $C^-(X_n, X_1, \cdots, X_{n-1}, X_0) = X_n^- > X_0^- = C^-(X_0, X_1, \cdots, X_n)$, which contradicts anonymity. This proves $G_{k,2}^- \leq G_{k+1,1}^-$ for all $k$. 
    
    Suppose $G_{k,2}^- < G_{k+1,1}^-$ for some $k$. Consider a vector of intervals $(X_0, \cdots, X_n)$ such that $X_0 = \{G_{k,2}^-\}$, $X_i=\{A^-\}$ for all $i=1,\cdots,n-k$, $X_j= \{A^+\}$ for all $j=n-k+1,\cdots,n-1$, and $X_n = \{G_{k+1,1}^-\}$. 
    We have $F_k^-(X_0) = \med(G_{k,2}^-, G_{k,1}^-, G_{k,2}^-) = G_{k,2}^- = X_0^-$ and $F_{k+1}^-(X_0) = \med(G_{k,2}^-, G_{k+1,1}^-, G_{k+1,2}^-) = G_{k+1,1}^- = X_n^-$. Combining $F_1(X_0) \leq \cdots \leq F_{n+1}(X_0)$, we have $F_1^-(X_0) \leq \cdots \leq F_k^-(X_0) = X_0^- = G_{k,2}^- < G_{k+1,1}^- = X_n^- = F_{k+1}^-(X_0) \leq \cdots \leq F_{n+1}^-(X_0)$. Then (\ref{eq:ind105}) implies
\begin{equation*}
    C^-(X_0, X_1, \cdots, X_n)
    = 
    \med \Bigg(\begin{array}{l}
        \overbrace{A^-, \cdots, A^-}^{n-k}, \overbrace{A^+, \cdots, A^+}^{k-1}, X_n^-, \\
        \underbrace{F_1^-(X_0), \cdots, F_{k-1}^-(X_0)}_{k-1}, X_0^-, X_n^-, \underbrace{F_{k+2}^-(X_0), \cdots, F_{n+1}^-(X_0)}_{n-k}
    \end{array}
    \Bigg) 
    =
    X_n^-.
\end{equation*}
    Consider a permutation exchanging $0$ and $n$. We get $C^-(X_n, X_1, \cdots, X_{n-1}, X_0) = X_0^- < X_n^- = C^-(X_0, X_1, \cdots, X_n)$, contradicting anonymity. Thus $G_{k,2}^- \geq G_{k+1,1}^-$ for all $k$. 

    Combining $G_{k,2}^- \geq G_{k+1,1}^-$ and $G_{k,2}^- \leq G_{k+1,1}^-$, we have shown the left claim $G_{k,2}^- = G_{k+1,1}^-$ for all $k$. The right claim $G_{k,2}^+ = G_{k+1,1}^+$ can be argued in a similar way. So $G_{k,2} = G_{k+1,1}$ for all $k=1, \cdots,n$. 
    Let $D_k = G_{k,1}$ for all $k = 1, \cdots, n+1$, and $D_{n+2} = G_{n+1,2}$. Rewrite (\ref{eq:ind104}) as
\begin{equation*}
    C(X_0, X_1, \cdots, X_n) 
    = 
    \medI(X_1, \cdots, X_n, \medI (X_0, D_1, D_2), \cdots, \medI (X_0, D_{n+1}, D_{n+2})), 
\end{equation*}
    where $D_1\leq\cdots\leq D_{n+2}$.

    Finally, we claim that for each profile $(X_0, X_1, \cdots, X_n)\in\calI^{n+1}$ and each sequence of intervals $D_1, \cdots, D_{n+2}\in\calI$ such that $D_1\leq \cdots \leq D_{n+2}$, 
\begin{equation*}
    \medI (X_1, \cdots, X_n, \medI (X_0, D_1, D_2), \cdots, \medI (X_0, D_{n+1}, D_{n+2})) 
    = 
    \medI (X_0, \cdots, X_n, D_1, \cdots, D_{n+2}).
\end{equation*}
    We decompose it into left and right claims. The left claim is that for each $(X_0, X_1, \cdots, X_n)\in\calI^{n+1}$ and each sequence $D_1, \cdots, D_{n+2}\in\calI$ such that $D_1\leq \cdots \leq D_{n+2}$, 
\begin{equation*}
    \med (X_1^-, \cdots, X_n^-, \med(X_0^-, D_1^-, D_2^-), \cdots, \med(X_0^-, D_{n+1}^-, D_{n+2}^-)) \\
    =
    \med(X_0^-, \cdots, X_n^-, D_1^-, \cdots, D_{n+2}^-). 
\end{equation*}
    Notice that $(X_1^-, \cdots, X_n^-)\in A^n$ is arbitrary and the sequence $D_1^-, \cdots, D_{n+2}^-\in A$ is non-decreasing. The above equation is directly derived from Lemma \ref{lem:ind4}. Note that the right claim can be similarly defined (by replacing all $[\cdot]^-$ with $[\cdot]^+$) and implied by Lemma \ref{lem:ind4}. The claim is thus proved.
\end{proof}

\subsection{Unanimity, ontoness, and the restricted generalized median scheme}

We next define unanimity and ontoness. 

\begin{axiom}[Unanimity]
    Fix $N$. If $X_i = V\in\calI$ for all $i\in N$, then $C(X) = V$.
\end{axiom}

\begin{axiom}[Ontoness]
    Fix $N$. For all $V\in\calI$, there exists $X\in\calI$ such that $C(X) = V$.
\end{axiom}

Unanimity says that the social choice must respect the unanimous agreement. Ontoness, also called voter sovereignty, requires that the voting scheme is a surjective function. These two axioms are closely related.

\begin{lemma}\label{lem:ind9}
    Unanimity implies ontoness; ontoness plus strategy-proofness implies unanimity.
\end{lemma}

\begin{proof}
    Fix $N$. Suppose the voting scheme $C:\calI^N\to\calI$ satisfies unanimity. For all $V\in\calI$, we choose the peak profile $X = (V, \cdots, V)$, and by unanimity we obtain $C(X) = V$. Thus it verifies ontoness. 

    Suppose the voting scheme $C$ satisfies strategy-proofness and ontoness. For all $V\in\calI$, from ontoness, we can find a profile $X$ such that $C(X) = V$. For all $i\in N$ such that $X_i\neq V$, by strategy-proofness, we get $C(V, X_{-i})\in B(V, C(X)) = B(V, V) = \{V\}$, and thus $C(V, X_{-i}) = V$. Sequentially manipulating individual reports into $V$, we get $C(V, \cdots, V) = V$, as desired.
\end{proof}

Either unanimity or ontoness, combining anonymity and strategy-proofness, characterizes the restricted generalized median scheme.

\begin{theorem}\label{thm:ind4}
    The restricted generalized median scheme is the only anonymous and strategy-proof voting scheme satisfying either unanimity or ontoness.
\end{theorem}

\begin{proof}
    
    Fix $n$. The restricted generalized median scheme among $n$ agents is a subclass of $S_n$, and thus it satisfies anonymity and strategy-proofness. 
    
    We show unanimity. Note that the scheme contains $n$ agents and $n-1$ fixed ballots, where agents are majority. It implies that $\min_{i\in N} X_i^- \leq C^-(X) \leq \max_{i\in N} X_i^-$ and $\min_{i\in N} X_i^+ \leq C^+(X) \leq \max_{i\in N} X_i^+$ for all $X\in\calI^N$. If $X_i = V\in\calI$ for all $i\in N$, we get $\min_{i\in N} X_i^- =\max_{i\in N} X_i^- = V^-$ and $\min_{i\in N} X_i^+ =\max_{i\in N} X_i^+ = V^+$. Thus $C^-(X) = V^-$ and $C^+(X) = V^+$. That is, $C(X) = V$, as desired.

    Note that, by Lemma~\ref{lem:ind9}, unanimity and ontoness are equivalent under strategy-proofness. It also verifies the ontoness.

    Conversely, we show that strategy-proofness, anonymity, and unanimity implies the restricted generalized median scheme.
    
    Fix $N$. Suppose a voting scheme satisfies strategy-proofness, anonymity, and unanimity. From Theorem \ref{thm:ind1}, strategy-proofness and anonymity imply the generalized median interval voting scheme, i.e., $S_n$. 
    
    We show that at least one fixed ballot has a leftmost alternative $A^-$.
    Suppose no fixed ballot's leftmost alternative is $A^-$. That is, for each fixed ballot $D_j\in\calI$, there is $D_j^- > A^-$, where $j = 1, \cdots, n+1$. Consider a peak profile $X$ such that $X_i = \{A^-\}$ for all $i\in N$. Since the fixed ballots form a majority, we know $C^-(X) \geq \min_{j=1, \cdots, n+1} D_j^- > A^-$. Therefore $C(X)\neq \{A^-\}$, contradicting unanimity. 

    We show that at least one fixed ballot has leftmost alternative $A^+$.
    Suppose there is no such fixed ballot, i.e., $D_j^- < A^+$ for each fixed ballot $D_j\in\calI$, where $i = 1, \cdots, n+1$. Consider a peak profile $X$ with $X_i = \{A^+\}$ for all $i\in N$. Since the fixed ballots form a majority, we know $C^-(X)<A^+$. But unanimity implies $C(X) = \{A^+\}$, which leads to a contradiction.

    Similarly, we obtain the same claim for the rightmost alternatives of fixed ballots. Since $S_n$ has a median structure, we can organize the leftmost and rightmost alternatives of fixed ballots and obtain two fixed ballots $\{A^-\}$ and $\{A^+\}$, as desired. 
\end{proof}

\subsection{Efficiency and the median scheme}

Now we introduce efficiency and the tools used to analyze it. 

\begin{axiom}[Efficiency]
    Fix $N$. For all $\succsim\in\mathcal{P}^N$ with peaks $X\in\calI^N$, there is no interval $V\in \calI$ such that $V \succ_i C(X)$ for some $i\in N$, and $V \succsim_j C(X)$ for all $j\in N$.
\end{axiom}

Efficiency requires the social choice to be Pareto optimal. That is, there is no choice of interval that benefits at least one agent and doesn't harm any agent.

\begin{lemma}\label{lem:ind12}
    Efficiency implies unanimity and ontoness.
\end{lemma}

\begin{proof}
    We show efficiency implies unanimity. Fix $N$. Consider a profile $X = (V, \cdots, V)$, and any efficient voting scheme $C:\calI^N\to\calI$. Suppose $C(V, \cdots, V) \neq V$. Then for all $i\in N$, by single-peakedness of preference, there is $V\succ_i C(X)$, contradicting efficiency. So unanimity must hold.
    From Lemma~\ref{lem:ind9}, ontoness is then implied. 
\end{proof}

In order to introduce another geometric property, we use the following terminologies. 
Denote by $q_t(X)$ the set of feasible intervals in the $t$-th quadrant centered at $X\in\calI$; and by $\hat{q}_t(X)$ those in the $t$-th extended quadrant, i.e., the $t$-th quadrant along with adjacent semi-axes. For example, $q_1(X) = \{Y\in \calI\mid Y^- > X^-, Y^+ > X^+\}$, and $\hat{q}_1(X) = \{Y\in \calI\mid Y^- \geq X^-, Y^+ \geq X^+\}$. We provide an illustration of the (extended) quadrant in Figure~\ref{fig:ind1}.
We use $t+$ to represent the index of the quadrant opposite to $q_t$.

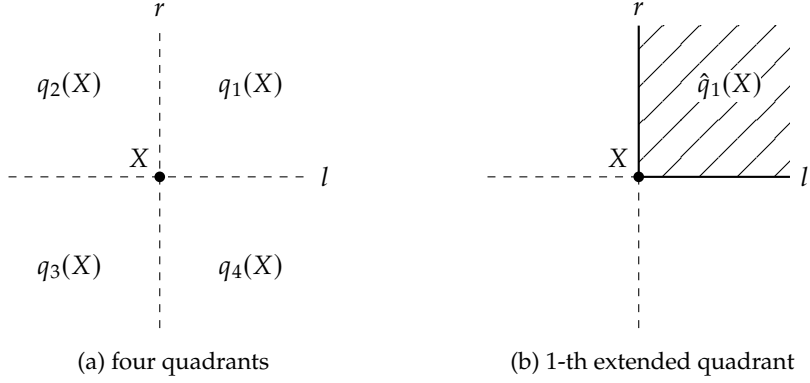
\begin{figure}
    \centering
\begin{subfigure}[t]{0.45\textwidth}
\centering
\begin{tikzpicture}[scale=1] 
    \draw[dashed] (-2,0) -- (2,0) node[right] {$l$}; 
    \draw[dashed] (0,-2) -- (0,2) node[above] {$r$}; 

    \fill (0,0) circle (2pt); \node[above left] at (0,0) {$X$}; 

    \node at (1.2, 1.2) {$q_1(X)$}; 
    \node at (-1.2, 1.2) {$q_2(X)$}; 
    \node at (-1.2, -1.2) {$q_3(X)$}; 
    \node at (1.2, -1.2) {$q_4(X)$}; 

\end{tikzpicture}
\subcaption{four quadrants}
\end{subfigure}
\begin{subfigure}[t]{0.45\textwidth}
\centering
\begin{tikzpicture}[scale=1] 
    \fill[ 
        pattern = {Lines[angle=45, distance=10pt, line width=0.4pt]}, 
        pattern color = black 
    ] (0,0) rectangle (2,2);

    \draw[dashed] (-2,0) -- (0,0);
    \draw[-, thick] (0,0) -- (2,0) node[right] {$l$}; 
    \draw[dashed] (0,-2) -- (0,0);
    \draw[-, thick] (0,0) -- (0,2) node[above] {$r$}; 

    \fill (0,0) circle (2pt); \node[above left] at (0,0) {$X$}; 

    \node[ fill=white, draw=none, inner xsep=12pt, inner ysep=6pt ] at (1.2, 1.2) {};
    \node at (1.2, 1.2) {$\hat{q}_1(X)$}; 
\end{tikzpicture}
\subcaption{1-th extended quadrant}
\end{subfigure}
    \caption{Quadrant and extended quadrant.}
    \label{fig:ind1}
\end{figure}

Motivated by \cite{block2010}, we formally propose the following geometric property.

\begin{axiom}[Quad-centrality]
    Fix $N$. For all $X\in\calI^N$ and all $t \in[4]$, there exists $i\in N$ such that $X_i\in\hat{q}_t(C(X))$.
\end{axiom}

This property says that there exists at least one peak in each extended quadrant of the social choice. 
Figure~\ref{fig:ind2} provides several examples of choices satisfying quad-centrality. Note that this property may be satisfied by less than 4 peaks.

It provides a sufficient condition for efficiency.

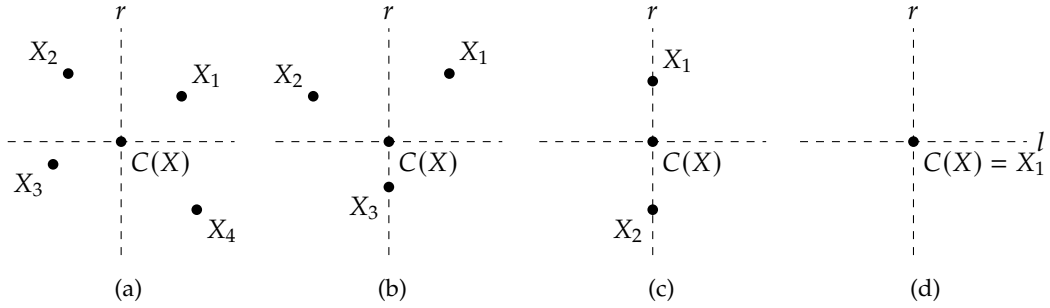
\begin{figure}
    \centering
\begin{subfigure}[t]{0.24\textwidth}
\centering
\begin{tikzpicture}[scale=1] 
    \draw[dashed] (-1.5,0) -- (1.5,0) node[right] {}; 
    \draw[dashed] (0,-1.5) -- (0,1.5) node[above] {$r$}; 

    \fill (0,0) circle (2pt); \node[below right] at (0,0) {$C(X)$}; 

    \fill (0.8, 0.6) circle (2pt); \node[above right] at (0.8, 0.6) {$X_1$}; 
    \fill (-0.7,0.9) circle (2pt); \node[above left] at (-0.7,0.9) {$X_2$}; 
    \fill (-0.9, -0.3) circle (2pt); \node[below left] at (-0.9, -0.3) {$X_3$}; 
    \fill (1,-0.9) circle (2pt); \node[below right] at (1,-0.9) {$X_4$}; 
\end{tikzpicture}
\subcaption{}
\end{subfigure}
\begin{subfigure}[t]{0.25\textwidth}
\centering
\begin{tikzpicture}[scale=1] 
    \draw[dashed] (-1.5,0) -- (1.5,0) node[right] {}; 
    \draw[dashed] (0,-1.5) -- (0,1.5) node[above] {$r$}; 

    \fill (0,0) circle (2pt); \node[below right] at (0,0) {$C(X)$}; 

    \fill (0.8, 0.9) circle (2pt); \node[above right] at (0.8, 0.9) {$X_1$}; 
    \fill (-1,0.6) circle (2pt); \node[above left] at (-1,0.6) {$X_2$}; 
    \fill (0, -0.6) circle (2pt); \node[below left] at (0, -0.6) {$X_3$}; 
\end{tikzpicture}
\subcaption{}
\end{subfigure}
\begin{subfigure}[t]{0.25\textwidth}
\centering
\begin{tikzpicture}[scale=1] 
    \draw[dashed] (-1.5,0) -- (1.5,0) node[right] {}; 
    \draw[dashed] (0,-1.5) -- (0,1.5) node[above] {$r$}; 

    \fill (0,0) circle (2pt); \node[below right] at (0,0) {$C(X)$}; 

    \fill (0, 0.8) circle (2pt); \node[above right] at (0, 0.8) {$X_1$}; 
    \fill (0, -0.9) circle (2pt); \node[below left] at (0, -0.9) {$X_2$}; 
\end{tikzpicture}
\subcaption{}
\end{subfigure}
\begin{subfigure}[t]{0.24\textwidth}
\centering
\begin{tikzpicture}[scale=1] 
    \draw[dashed] (-1.5,0) -- (1.5,0) node[right] {$l$}; 
    \draw[dashed] (0,-1.5) -- (0,1.5) node[above] {$r$}; 

    \fill (0,0) circle (2pt); \node[below right] at (0,0) {$C(X)=X_1$}; 
\end{tikzpicture}
\subcaption{}
\end{subfigure}

    \caption{Examples of social choice satisfying quad-centrality.}
    \label{fig:ind2}
\end{figure}

\begin{lemma}\label{lem:ind17}
    Quad-centrality implies efficiency.
\end{lemma}

\begin{proof}
    Fix $N$. Suppose a voting scheme $C:\calI^N\to\calI$ satisfies quad-centrality. That is, for peaks $X\in\calI^N$ and all $t \in[4]$, there exists $i\in N$ such that $X_i\in\hat{q}_t(C(X))$. Suppose to the contrary there exists an interval $V\neq C(X)$ such that $V \succ_j C(X)$ for some $j\in N$, and $V \succsim_k C(X)$ for all $k\in N$. Consider a deviation $V\in \hat{q}_t(C(X))$ for some $t$. Then there exists $k\in N$ such that $C(X)\in B(X_k, V)$. By single-peakedness, it implies $C(X)\succ_k V$. It yields the desired contradiction, and thus verifies efficiency.
\end{proof}

We also provide the following lemma as a tool to construct a contradiction to efficiency.

\begin{lemma}\label{lem:ind14}
    Fix $N$ and peaks $X\in\calI^N$. Suppose for some $t$, there is $X_i\notin \hat{q}_t(C(X))$ for all $i\in N$, and $q_{t+}(C(X))\neq \varnothing$. Then we can find an interval $V\in q_{t+}(C(X))$ and a preference profile $\succsim\in\mathcal{P}^N$ with peaks $X$ such that $V\succ_i C(X)$ for all $i\in N$.
\end{lemma}

\begin{proof}
    Suppose, without loss of generality, there is no peak in $\hat{q}_1(C(X))$, and consider a deviation $V\in q_3(C(X))$. Notice that the social choice $C(X)$ is never inside the betweenness set of the deviation $V$, satisfying $V^- < C^-(X)$ and $V^+ < C^+(X)$, and each peak $X_i$, which has $X_i^- < C^-(X)$ or $X_i^+ < C^+(X)$. Figure~\ref{fig:ind3} provides an illustration of betweenness sets. Then, by richness condition, we can find a preference profile such that all agents prefer the deviation $V$ to the choice $C(X)$. 
\end{proof}

Using this lemma, we can eliminate all inefficient schemes within the class of restricted generalized median schemes in the following two steps.

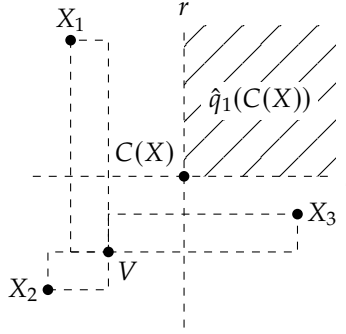
\begin{figure}
    \centering
\begin{tikzpicture}[scale=1] 
    \fill[ 
        pattern = {Lines[angle=45, distance=10pt, line width=0.4pt]}, 
        pattern color = black 
    ] (0,0) rectangle (2,2);

    \node[ fill=white, draw=none, inner xsep=22pt, inner ysep=7pt ] at (1, 1) {};
    \node at (1, 1) {$\hat{q}_1(C(X))$}; 

    \draw[dashed] (-2,0) -- (2,0) node[right] {$l$}; 
    \draw[dashed] (0,-2) -- (0,2) node[above] {$r$}; 

    \fill (0,0) circle (2pt); \node[above left] at (0,0) {$C(X)$}; 

    \coordinate (V) at (-1,-1);
    \fill (-1,-1) circle (2pt); \node[below right] at (-1,-1) {$V$}; 

    \coordinate (X1) at (-1.5,1.8);
    \fill (-1.5,1.8) circle (2pt); \node[above]  at (-1.5,1.8) {$X_1$};
    \coordinate (X2) at (-1.8,-1.5);
    \fill (-1.8,-1.5) circle (2pt); \node[left] at (-1.8,-1.5) {$X_2$}; 
    \coordinate (X3) at (1.5,-0.5);
    \fill (1.5,-0.5) circle (2pt); \node[ right] at (1.5,-0.5) {$X_3$}; 

    \node[ draw, dashed,  inner sep=0pt, fit=(V)(X1) ] {};
    \node[ draw, dashed,  inner sep=0pt, fit=(V)(X2) ] {};
    \node[ draw, dashed,  inner sep=0pt, fit=(V)(X3) ] {};

\end{tikzpicture}
    \caption{Betweenness sets of peaks $X$ out of $\hat{q}_1$ and a deviation $V$ in $q_3$.}
    \label{fig:ind3}
\end{figure}

\begin{lemma}\label{lem:ind16}
    A voting scheme satisfying efficiency, anonymity, and strategy-proofness belongs to the rank-based scheme.
\end{lemma}

\begin{proof}
    Fix $N$. Suppose a voting scheme $C:\calI^N\to\calI$ satisfies efficiency, anonymity, and strategy-proofness. By Lemma~\ref{lem:ind12}, efficiency implies unanimity. Then, from Theorem~\ref{thm:ind4}, we get $C$ is a restricted generalized median scheme. Using efficiency, we show that all fixed ballots of $C$ are either $\{A^-\}$ or $\{A^+\}$. 

    Let $D_1, \cdots, D_{n-1}$ denote fixed ballots. Without any loss, we assume $D_1\leq\cdots\leq D_{n-1}$. Suppose $D_k\notin\{\{A^-\}, \{A^+\}\}$ for some $k$. Let $X_i = \{A^-\}$ for all $i = 1, \cdots, n-k$, and $X_j = \{A^+\}$ for all $j = n-k+1, \cdots, n$. From Lemma~\ref{lem:ind3}, we get $C(X) = D_k$.

    Case 1: $D_k^- = D_k^+ \notin \{A^-, A^+\}$. Note that the 4-th extended quadrant centered by $C(X) = D_k$ only contains $D_k$. That is, $\hat{q}_4(C(X)) = \{D_k\}$, where $X_i\neq D_k$ for all $i\in N$. Note that we can always find a feasible interval $A$ in the 2-th quadrant. That is, $q_2(C(X))\neq \varnothing$. By Lemma~\ref{lem:ind14}, we can find a preference $\succsim\in\mathcal{P}$ and an interval $V\in q_2(C(X))$ such that $V\succ_i C(X)$ for all $i\in N$, which contradicts efficiency. 

    Case 2: $D_k^- < D_k^+$. It implies that $D_k^-\neq A^+$ and $D_k^+\neq A^-$, and thus $\hat{q}_2(C(X)) = \{[l, r]\in\calI\mid l\leq D_k^-, r\geq D_k^+\}$ does not contain $\{A^-\}$ and $\{A^+\}$. Therefore, $X_i\notin \hat{q}_2(C(X))$ for all $i\in N$. Since $A$ is a real interval and hence dense, we can always find an alternative $a\in A$ such that $D_k^- < a < D_k^+$. It implies $[a, a]\in q_4(C(X))$. By Lemma~\ref{lem:ind14}, we can find a preference $\succsim\in\mathcal{P}$ and an interval $V\in q_4(C(X))$, e.g., $[a, a]$, such that $V\succ_i C(X)$ for all $i\in N$, contradicting efficiency.     

    We conclude that each fixed ballot is either $\{A^-\}$ or $\{A^+\}$, and thus $C$ belongs to the rank-based scheme.
\end{proof}

\begin{lemma}\label{lem:ind15}
    For an odd $n$, the median scheme is the only efficient rank-based scheme. For an even $n$, no rank-based scheme is efficient.
\end{lemma}

\begin{proof}
    For $n=1$, the only rank-based scheme outputs the median scheme. 

    Fix an odd $n\geq 3$. Suppose $n = 2m+1$. Note that we can always find an alternative $a\in A$ such that $A^- < a < A^+$. 

    We show that the $\gamma$-th rank-based scheme violates efficiency if $\gamma \leq m$ or $\gamma\geq m+2$.

    Case 1: $\gamma\leq m$. Let $X_i = A$ for all $i = 1, \cdots, n-\gamma$, and $X_j = \{a\}$ for all $j = n-\gamma+1, \cdots, n$. For the leftmost alternatives, there are $n-\gamma$ copies of $A^-$ and $\gamma$ copies of $a$; for the rightmost alternatives, there are $\gamma$ copies of $a$ and $n-\gamma$ copies of $A^+$. Note that $n-\gamma\geq \gamma$.
    The $\gamma$-th rank-based scheme gives that $C(X) = [A^-, a]$. Note that $\hat{q}_3(C(X)) = \{[l, r]\in\calI\mid l\leq A^-, r\leq a\}$, and $X_i \notin \hat{q}_3(C(X))$ for all $i\in N$. Also notice that $[a, A^+]$ is feasible. From Lemma~\ref{lem:ind14}, we can find a preference profile $\succsim\in\mathcal{P}^N$ and an interval $V\in q_1(C(X))$, e.g., $[a, A^+]$, such that $V\succ_i C(X)$ for all $i\in N$.

    Case 2: $\gamma\geq m+2$. Similarly, let $X_i = A$ for all $i = 1, \cdots, n-\gamma+1$, and $X_j = \{a\}$ for all $j = n-\gamma+2, \cdots, n$. We get the $\gamma$-th rank-based scheme yields $C(X) = [a, A^+]$. Notice that for all $i\in N$, there is $X_i\notin\hat{q}_1(C(X)) = \{[l, r]\in\calI\mid l\geq a, r\geq A^+\}$, and $[A^-, a]\in q_3(C(X))$ is always feasible. By Lemma~\ref{lem:ind14}, we can find a preference profile $\succsim\in\mathcal{P}^N$ and an interval $V\in q_3(C(X))$, e.g., $[A^-, a]$, such that $V\succ_i C(X)$ for all $i\in N$.

    We show that the median scheme, i.e., the $(m+1)$-th rank-based scheme, is efficient by verifying quad-centrality. That is to show for all $t$, there exists $X_i\in \hat{q}_t(C(X))$ for some $i\in N$. 
    Suppose to the contrary, without loss of generality, there is $X_i\notin \hat{q}_3(C(X))$ for all $i\in N$. That is, $X_i^- > C^-(X)$ or $X_i^+ > C^+(X)$ for all $i\in N$. It implies $|\{i\in N\mid X_i^- > C^-(X)\}| + |\{i\in N\mid X_i^+ > C^+(X)\}|\geq n$. Note that $C$ is a median scheme. Then we have $|\{i\in N\mid X_i^- > C^-(X)\}| \leq m$ and $|\{i\in N\mid X_i^+ > C^+(X)\}| \leq m$. It implies that $n\leq |\{i\in N\mid X_i^- > C^-(X)\}| + |\{i\in N\mid X_i^+ > C^+(X)\}| \leq 2m < n$, a contradiction. Therefore $C$ satisfies quad-centrality, and by Lemma~\ref{lem:ind17}, it is efficient.

    Therefore, for an odd $n$, the median scheme is the only efficient rank-based scheme.

    For each even $n\geq 2$. Suppose $n = 2m$. We consider similar constructions. 

    Case 1': $\gamma\leq m$. Let $X_i = A$ for all $i = 1, \cdots, n-\gamma$, and $X_j = \{a\}$ for all $j = n-\gamma+1, \cdots, n$. We obtain $C(X) = [A^-, a]$. From Lemma~\ref{lem:ind14}, we can find a preference profile $\succsim\in\mathcal{P}^N$ and an interval $V=[a, A^+]$, such that $V\succ_i C(X)$ for all $i\in N$.

    Case 2': $\gamma\geq m+1$. Consider $X_i = A$ for all $i = 1, \cdots, n-\gamma+1$, and $X_j = \{a\}$ for all $j = n-\gamma+2, \cdots, n$. We get $C(X) = [a, A^+]$. By Lemma~\ref{lem:ind14}, we can find a preference profile $\succsim\in\mathcal{P}^N$ and an interval $V\in q_3(C(X))$, e.g., $[A^-, a]$, such that $V\succ_i C(X)$ for all $i\in N$.

    Therefore all rank-based schemes for an even $n$ violate efficiency.
\end{proof}

The preceding steps are sufficient to show the characterization of the median scheme. 

\begin{theorem}\label{thm:ind3}
    For an odd $n$, the median scheme is the only efficient, anonymous, and strategy-proof voting scheme. 
\end{theorem}

\begin{proof}
    Fix an odd $n$. Suppose the voting scheme $C:\calI^n\to\calI$ satisfies efficiency, anonymity, and strategy-proofness. 
    By Lemma~\ref{lem:ind12}, efficiency implies unanimity. From Theorem~\ref{thm:ind4}, unanimity, anonymity, and strategy-proofness imply that $C$ is a restricted generalized median scheme. 
    From Lemma~\ref{lem:ind16}, efficiency, anonymity, and strategy-proofness imply that $C$ is a rank-based scheme.
    From Lemma~\ref{lem:ind15}, we know $C$ must be a median scheme.
    
    Conversely, Lemma~\ref{lem:ind15} shows that the median scheme is the only efficient rank-based scheme. Note that each rank-based scheme is a generalized median scheme, which satisfies anonymity and strategy-proofness. This verifies efficiency, anonymity, and strategy-proofness of the median scheme.
\end{proof}

\section{Concluding remarks}\label{sec:ind5}

This paperstudies voting schemes, and characterizes several classes of strategy-proof voting schemes under single-peaked preferences over intervals. The leftright scheme, which generalizes the minmax scheme \citep{moulin1980}, is the only strategy-proof voting scheme (Theorem~\ref{thm:ind5}). All schemes considered in this paper belong to the leftright scheme. Adding anonymity, it becomes the generalized median scheme (Theorem~\ref{thm:ind1}); and this scheme is restricted further by adding unanimity or ontoness (Theorem~\ref{thm:ind4}). The results for the (restricted) generalized median schemes are used for the characterization of the median scheme. We find that the median scheme is the only efficient, anonymous, and strategy-proof voting scheme (Theorem~\ref{thm:ind3}). Recall that three axioms characterize the restricted generalized median scheme for one-dimensional model \citep{moulin1980}. Our result reveals the difficulty of preserving efficiency and strategy-proofness in two-dimensional social choice. 

In addition to the analysis above, we provide several discussions. 

\subsection{Alternative set}
Similar to \cite{klaus2020}, the alternative set in our paper is a closed interval on $\mathbb{R}$. In contrast, \cite*{barbera1993} consider a discrete setting, i.e., integer interval, which is different from ours. We claim that most of our results still hold in the discrete model, while Lemma~\ref{lem:ind16}, and thus Theorem~\ref{thm:ind5}, may fail. 

Lemma~\ref{lem:ind16} works only when the alternative set $A$ is dense. An alternative set $A$ is dense if for all distinct $a,b\in A$, there exists $c\in A$ such that $a < c < b$. Note that the alternative set as a closed interval is dense.

Recall the case 2 of the proof. Whenever a non-extreme fixed ballot $D_k$ is chosen, we can always find an alternative $c\in A$ such that $D_k^- < a < D_k^+$, construct a Pareto improvement to $[a, a]$, and then eliminate it. In contrast, suppose the alternative set $A$ is an integer interval with at least 3 alternatives; and consider a non-extreme fixed ballot $D_k = [A^+-1, A^+]$. We do not have any feasible integer interval in the 4-th quadrant. Thus the improvement cannot be constructed. We illustrate this comparison in Figure~\ref{fig:ind4}.

\begin{figure}
    \centering
\begin{subfigure}[t]{0.45\textwidth}
\centering
\begin{tikzpicture}[scale=1] 
    \draw[dashed] (-2.8,0)--(2,0)  node[below] {$l$}; 
    \draw[dashed] (0,-3) -- (0,1.2) node[above] {$r$}; 

    \fill (0,0) circle (2pt); \node[above left] at (0,0) {$D_k$}; 
    \fill (1,0) circle (2pt); \node[above right] at (1,0) {$X_1=\{A^+\}$}; 
    \fill (0,-1) circle (2pt); \node[left] at (0,-1) {$\{A^+-1\}$}; 

    \draw[dashed] (1, 0) -- (-2, -3) ;
    \fill (-2, -3) circle (2pt); \node[left] at (-2, -3) {$X_2=\{A^-\}$}; 
    \draw[dashed] (-2, -3) -- (-2, 0) ;
    \fill (-2, 0) circle (2pt); \node[above left] at (-2, 0) {$A$}; 
    
    \node at (-1.8, 1.2) {$\forall i: X_i\notin \hat{q}_2(D_k)$}; 
    \node at (1.8, -1.2) {$q_4(D_k)=\varnothing$}; 
\end{tikzpicture}
\subcaption{discrete alternative set}
\end{subfigure}
\begin{subfigure}[t]{0.45\textwidth}
\centering
\begin{tikzpicture}[scale=1] 
    \draw[dashed] (-2.8,0)--(2,0)  node[below] {$l$}; 
    \draw[dashed] (0,-3) -- (0,1.2) node[above] {$r$}; 

    \fill (0,0) circle (2pt); \node[above left] at (0,0) {$D_k$}; 
    \fill (1,0) circle (2pt); \node[above right] at (1,0) {$X_1=\{A^+\}$}; 
    \fill (0,-1) circle (2pt); \node[left] at (0,-1) {$\{A^+-1\}$}; 

    \draw[-] (1, 0) -- (-2, -3) ;
    \fill (-2, -3) circle (2pt); \node[left] at (-2, -3) {$X_2=\{A^-\}$}; 
    \draw[-] (-2, -3) -- (-2, 0) ;
    \fill (-2, 0) circle (2pt); \node[above left] at (-2, 0) {$A$}; 
    \draw[-] (-2, 0) -- (1,0) ;

    \fill (0.5, -0.5) circle (2pt); \node[right] at (0.5, -0.5) {$[a,a]$}; 
    
    \node at (-1.8, 1.2) {$\forall i: X_i\notin \hat{q}_2(D_k)$}; 
    \node at (1.8, -1.2) {$q_4(D_k)\neq\varnothing$}; 
\end{tikzpicture}
\subcaption{continuous alternative set}
\end{subfigure}
    \caption{Lemma~\ref{lem:ind16} fails in the discrete model.}
    \label{fig:ind4}
\end{figure}
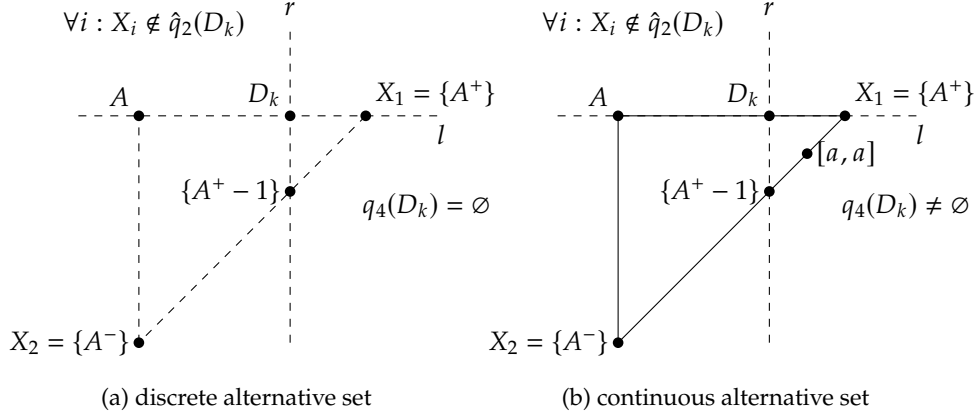

\subsection{Social choice function}

In this paper, we simplify our analysis into the voting schemes, i.e., peak-only social choice functions. Due to this simplification, each agent only needs to report peak interval instead of full preference. As explained before, this simplification is natural and practical. But the cost is that we actually assume an additional axiom, namely peak-onlyness, of the social choice function. 

Fix $N$. A social choice function is a mapping $(\succsim_1, \cdots, \succsim_n)\mapsto F(\succsim_1, \cdots, \succsim_n)$ from the agents' preference profile, rather than their peak profile, to a choice of interval. Let $P(\succsim_i)$ denote the peak interval of preference $\succsim_i$ of agent $i$. The peak-onlyness is defined as follows:

\begin{axiom}[Peak-onlyness]
    Fix $N$. For all $\succsim, \succsim'\in\mathcal{P}^N$ such that $P(\succsim_i) = P(\succsim_i')$ for all $i\in N$, we have $F(\succsim) = F(\succsim')$.
\end{axiom}

This axiom says that the social choice depends only on the peaks of agents' preferences. 

Note that our set of feasible intervals, a triangle on $\mathbb{R}^2$, is non-Cartesian and compact, and the preference is assumed to be single-peaked. Adding either unanimity or ontoness, the range of the social choice function would coincide with the set of feasible intervals. From \cite*{barbera1998}, we know such social choice function must take the form of generalized median voter scheme, and thus peak-onlyness is implied. 

Since efficiency implies unanimity, we can strengthen our Theorem~\ref{thm:ind3} as follows.

\begin{theorem*}[\textbf{4'}]
    For an odd $n$, the median scheme is the only efficient, anonymous, and strategy-proof social choice function.
\end{theorem*}

\subsection{Group strategy-proofness}

Group strategy-proofness is a stronger notion of strategy-proofness. 
It says that no coalition of agents has an incentive to collectively misreport their peaks. We establish this axiom as follows:

\begin{axiom}[Group strategy-proofness]
    Fix $N$. For all $\succsim\in\mathcal{P}^N$ with peaks $X\in\calI^N$ and all $S\subseteq N$, there is no $Y_S\in\calI^S$ such that for all $i\in S$, $C(Y_S, X_{N\setminus S})\succ_i C(X)$.
\end{axiom}

In the one-dimensional model, \cite{moulin1980} shows that group strategy-proofness coincides with strategy-proofness given efficiency and anonymity. However, in our interval model, this result doesn't hold. From Theorem~\ref{thm:ind3}, we know that only the median scheme satisfies efficiency, anonymity, and strategy-proofness. To verify whether group strategy-proofness coincides with strategy-proofness, we only need to check whether the median scheme for an odd number of agents satisfies group strategy-proofness. Indeed, no voting scheme satisfies all three axioms. 
We provide a simple example to illustrate. 

\begin{example}
    The median scheme fails group strategy-proofness.

    Suppose there are three agents. Their true peaks are $X_1 = [1, 2]$, $X_2 = [2, 3]$, $X_3 = [1, 3]$. The median scheme gives $\medI(X_1, X_2, X_3) = [1, 3]$.
    Consider a deviation $Y_1 = Y_2 = [2, 2]$ of coalition $\{1, 2\}$. Then the manipulated median choice becomes $\medI(Y_1, Y_2, X_3) = [2, 2]$. 

    Suppose agent 1 has a utility function $u_1(V) = - |V^--1| - 2|V^+ - 2|$; and agent 2 has a utility function $u_2(V) = - 2|V^- - 2| - |V^+ - 3|$. Then we have $u_1([1, 3]) = -2 < -1 = u_1([2, 2])$; and $u_2([1, 3]) = -2 < -1 = u_2([2, 2])$. That is, agents in coalition $\{1, 2\}$ strictly benefit from misreporting. This violates group strategy-proofness. 
\end{example}

\begin{proposition}
    No voting scheme satisfies efficiency, anonymity, and group strategy-proofness.
\end{proposition}


\bibliography{main}

@article{barbera1991,
  title={Voting by committees},
  author={Barber{\`a}, Salvador and Sonnenschein, Hugo and Zhou, Lin},
  journal={Econometrica},
  volume={59},
  number={3},
  pages={595--609},
  year={1991},
  publisher={JSTOR}
}

@article{barbera1993,
  title={Generalized median voter schemes and committees},
  author={Barber{\`a}, Salvador and Gul, Faruk and Stacchetti, Ennio},
  journal={Journal of Economic Theory},
  volume={61},
  number={2},
  pages={262--289},
  year={1993},
  publisher={Elsevier}
}

@article{barbera1997,
  title={Voting under constraints},
  author={Barber{\`a}, Salvador and Mass{\'o}, Jordi and Neme, Alejandro},
  journal={Journal of Economic Theory},
  volume={76},
  number={2},
  pages={298--321},
  year={1997},
  publisher={Elsevier}
}

@article{barbera1998,
  title={Strategy-proof voting on compact ranges},
  author={Barber{\`a}, Salvador and Mass{\'o}, Jordi and Serizawa, Shigehiro},
  journal={Games and Economic Behavior},
  volume={25},
  number={2},
  pages={272--291},
  year={1998},
  publisher={Elsevier}
}

@article{barbera2001,
  title={Strategy-proof social choice correspondences},
  author={Barber{\`a}, Salvador and Dutta, Bhaskar and Sen, Arunava},
  journal={Journal of Economic Theory},
  volume={101},
  number={2},
  pages={374--394},
  year={2001},
  publisher={Elsevier}
}

@article{barbera2011,
  title={Strategyproof social choice},
  author={Barber{\`a}, Salvador},
  journal={Handbook of Social Choice and Welfare},
  volume={2},
  pages={731--831},
  year={2011},
  publisher={Elsevier}
}

@article{block2010,
  title={Efficient and strategy-proof voting over connected coalitions: A possibility result},
  author={Block, Veronica},
  journal={Economics Letters},
  volume={108},
  number={1},
  pages={1--3},
  year={2010},
  publisher={Elsevier}
}

@article{border1983,
  title={Straightforward elections, unanimity and phantom voters},
  author={Border, Kim C and Jordan, James S},
  journal={The Review of Economic Studies},
  volume={50},
  number={1},
  pages={153--170},
  year={1983},
  publisher={Wiley-Blackwell}
}

@article{clemen1989,
  title={Combining forecasts: A review and annotated bibliography},
  author={Clemen, Robert T},
  journal={International Journal of Forecasting},
  volume={5},
  number={4},
  pages={559--583},
  year={1989},
  publisher={Elsevier}
}

@article{farfel2011,
  title={Aggregating value ranges: Preference elicitation and truthfulness},
  author={Farfel, Joseph and Conitzer, Vincent},
  journal={Autonomous Agents and Multi-Agent Systems},
  volume={22},
  number={1},
  pages={127--150},
  year={2011},
  publisher={Springer}
}

@article{gibbard1973,
  title={Manipulation of voting schemes: A general result},
  author={Gibbard, Allan},
  journal={Econometrica},
  volume = {41},
  number = {4},
  pages={587--601},
  year={1973},
  publisher={JSTOR}
}

@article{kim1984,
  title={Nonmanipulability in two dimensions},
  author={Kim, Ki Hang and Roush, Fred W},
  journal={Mathematical Social Sciences},
  volume={8},
  number={1},
  pages={29--43},
  year={1984},
  publisher={Elsevier}
}

@article{klaus2020,
  title={On strategy-proofness and single-peakedness: median-voting over intervals},
  author={Klaus, Bettina and Protopapas, Panos},
  journal={International Journal of Game Theory},
  volume={49},
  number={4},
  pages={1059--1080},
  year={2020},
  publisher={Springer}
}

@unpublished{miller2025,
  author = {Miller, Alan D},
  title  = {The limits of tolerance},
  year   = {2025},
  note   = {{mimeo}. \href{https://doi.org/10.48550/arXiv.2501.00578}{arXiv:2501.00578}}
}

@article{moulin1980,
  title={On strategy-proofness and single peakedness},
  author={Moulin, Herv{\'e}},
  journal={Public Choice},
  volume={35},
  number={4},
  pages={437--455},
  year={1980},
  publisher={Springer}
}

@article{nehring2007,
  title={The structure of strategy-proof social choice—Part I: General characterization and possibility results on median spaces},
  author={Nehring, Klaus and Puppe, Clemens},
  journal={Journal of Economic Theory},
  volume={135},
  number={1},
  pages={269--305},
  year={2007},
  publisher={Elsevier}
}

@article{nehring2007b,
  title={Efficient and strategy-proof voting rules: A characterization},
  author={Nehring, Klaus and Puppe, Clemens},
  journal={Games and Economic Behavior},
  volume={59},
  number={1},
  pages={132--153},
  year={2007},
  publisher={Elsevier}
}

@article{peters1992,
  title={Pareto optimality, anonymity, and strategy-proofness in location problems},
  author={Peters, Hans and van der Stel, Hans and Storcken, Ton},
  journal={International Journal of Game Theory},
  volume={21},
  number={3},
  pages={221--235},
  year={1992},
  publisher={Springer}
}

@article{sprumont1995,
  title={Strategyproof collective choice in economic and political environments},
  author={Sprumont, Yves},
  journal={Canadian Journal of Economics},
  volume={28},
  number={1},
  pages={68--107},
  year={1995},
  publisher={JSTOR}
}

@article{satterthwaite1975,
  title={Strategy-proofness and Arrow's conditions: Existence and correspondence theorems for voting procedures and social welfare functions},
  author={Satterthwaite, Mark Allen},
  journal={Journal of Economic Theory},
  volume={10},
  number={2},
  pages={187--217},
  year={1975},
  publisher={Elsevier}
}
\end{document}